\documentclass[11pt]{article}
\usepackage{fullpage}
\usepackage{amsfonts,amsmath,amsthm,amssymb}
\usepackage{bbm}
\usepackage[utf8]{inputenc} 
\usepackage[T1]{fontenc}    
\usepackage{url}
\usepackage{color}
\usepackage[usenames,dvipsnames,svgnames,table]{xcolor}
\usepackage[colorlinks=true, linkcolor=red, urlcolor=blue, citecolor=gray, backref=page]{hyperref}
\usepackage{graphicx}
\usepackage{caption}
\usepackage{subcaption}
\usepackage{booktabs}  
\usepackage{microtype} 
\usepackage{xcolor} 

\usepackage{etoolbox}

\makeatletter

\newcommand{\define}[4][ignore]{%
  \ifstrequal{#1}{ignore}{}{
  \@namedef{thmtitle@#2}{#1}}%
  \@namedef{thm@#2}{#4}%
  \@namedef{thmtypen@#2}{lemma}%
  \newtheorem{thmtype@#2}[theorem]{#3}%
  \newtheorem*{thmtypealt@#2}{#3~\ref{#2}}%
}

\newcommand{\state}[1]{%
  \@namedef{curthm}{#1}
  \@ifundefined{thmtitle@#1}{
  \begin{thmtype@#1}
    }{
  \begin{thmtype@#1}[\@nameuse{thmtitle@#1}]
  }
    \label{#1}
    \@nameuse{thm@#1}
  \end{thmtype@#1}
  \@ifundefined{thmdone@#1}{
  \@namedef{thmdone@#1}{stated}%
  }{}
}

\newcommand{\restate}[1]{%
  \@namedef{curthm}{#1}
  \@ifundefined{thmtitle@#1}{
    \begin{thmtypealt@#1}
    }{
  \begin{thmtypealt@#1}[\@nameuse{thmtitle@#1}]
  }
    \@nameuse{thm@#1}
  \end{thmtypealt@#1}
  \@ifundefined{thmdone@#1}{
  \@namedef{thmdone@#1}{stated}%
  }{}
}

\newcommand{\thmlabel}[1]{
  \@ifundefined{thmdone@\@nameuse{curthm}}{\label{#1}
    }{\tag*{\eqref{#1}}}
}
\makeatother

\usepackage{tikz}
\usepackage[T1]{fontenc}

\usetikzlibrary{matrix,decorations.pathreplacing}

\usepackage{xcolor}
\usepackage[symbol]{footmisc}

\usepackage{easybmat}
\usepackage[normalem]{ulem}
\usepackage{multirow,bigdelim}
\usepackage{amsmath}
\usepackage{amsthm}
\usepackage{amssymb}
\usepackage{algorithm}
\usepackage{algpseudocode}
\usepackage{bbm}
\usepackage{xspace}
\newtheorem{theorem}{Theorem}
\newtheorem{fact}[theorem]{Fact}
\newtheorem{lemma}[theorem]{Lemma}

\newtheorem{claim}[theorem]{Claim}
\newtheorem{definition}[theorem]{Definition}
\newtheorem{corollary}[theorem]{Corollary}
\newcommand{\wh}{\widehat}
\newcommand{\wt}{\widetilde}

\newcommand{\eps}{\epsilon}

\renewcommand{\varepsilon}{\epsilon}
\renewcommand{\tilde}{\wt}
\renewcommand{\hat}{\wh}

\DeclareMathOperator*{\E}{{\bf{E}}}

\DeclareMathOperator{\poly}{poly}

\newcommand{\rEstimate }{\textsc{Refine-Estimate}\xspace}

\newcommand{\RR}{\mathcal{R}}

\newcommand{\highlight}[1]{{ #1}}
\newcommand{\support}[1]{\textup{supp}(#1)}
\newcommand{\Rbad}{\mathcal{R}_{\textup{bad}}}
\newcommand{\Vboundary}{V_{\textup{boundary}}}
\newcommand{\Rboundary}{\mathcal{R}_{\textup{boundary}}}
\newcommand{\sampling}{\textsc{Sampling}\xspace}
\newcommand{\uneighbor}{\textsc{Uniform-Neighbor}\xspace}
\newcommand{\estimator}{\textsc{Edge-Estimator}\xspace}
\newcommand{\approxdegree}{\textsc{Estimate-Degree}\xspace}

\newcommand{\boundary}[1]{\mu(#1)}
\newcommand{\nbralg}{\textsc{Neighborhood-Size}\xspace}

\newcommand{\nbrest}[1]{{\eta}^{#1}_{\text{est}}}
\newcommand{\nbr}[1]{\Gamma(#1)}
\newcommand{\bis}[2]{\mathcal{BIS}(#1, #2)}

\newcommand{\Bq}{\mathcal{BIS}}
\newcommand{\Or}{\mathcal{OR}}
\newcommand{\yes}{\text{`1'}\xspace}
\newcommand{\no}{\text{`0'}\xspace}
\newcommand{\Count}[1]{\texttt{count}(#1)}
\newcommand{\zz}[1]{z^{#1}}
\newcommand{\tmin}{t_{\textup{min}}}

\definecolor{mygreen}{RGB}{80,180,0}
\definecolor{b2}{RGB}{51,153,255}
\definecolor{mycy2}{RGB}{255,51,255}
\newcommand{\Gsup}{G^\text{sup}}
\newcommand{\Vsup}{V^\text{sup}}
\newcommand{\Esup}{E^\text{sup}}

\newcommand{\remove}[1]{{\iffalse #1 \fi}}
  \usepackage{nth}
  \usepackage{intcalc}

\title{Non-Adaptive Edge Counting and Sampling via Bipartite Independent Set Queries}

\author{Raghavendra Addanki\footnote{Most of the work done while a graduate student at the University of Massachusetts Amherst.}\\ Adobe Research \\ \texttt{raddanki@adobe.com} \and Andrew McGregor \\ University of Massachusetts Amherst \\ \texttt{mcgregor@cs.umass.edu} \and Cameron Musco \\ University of Massachusetts Amherst \\ \texttt{cmusco@cs.umass.edu}}

\begin{document}

\maketitle

\begin{abstract}
We study the problem of estimating the number of edges in an $n$-vertex graph, accessed via the \emph{Bipartite Independent Set} query model introduced by~Beame et al.~(ITCS '18). In this model, each query returns a Boolean, indicating the existence of at least one edge between two specified sets of nodes. We present a  \emph{non-adaptive} algorithm that  returns a $(1\pm \epsilon)$ relative error approximation to the number of edges, with query complexity $\tilde O(\epsilon^{-5}\log^{5} n )$, where $\tilde O(\cdot)$ hides $\poly(\log \log n)$ dependencies. This is the  first non-adaptive algorithm in this setting achieving $\poly(1/\epsilon,\log n)$ query complexity. Prior work requires  $\Omega(\log^2 n)$ rounds of adaptivity. We avoid this by taking a fundamentally different approach, inspired by work on single-pass streaming algorithms. 
Moreover, for constant $\epsilon$, our query complexity significantly improves on the  best known adaptive algorithm due to Bhattacharya et al. (STACS '22), which requires $O(\epsilon^{-2} \log^{11} n)$ queries.
Building on our edge estimation result, we give the first {non-adaptive} algorithm for outputting a nearly uniformly sampled edge with query complexity $\tilde O(\epsilon^{-6} \log^{6} n)$, improving on the works of Dell et al.~(SODA '20) and Bhattacharya et al.~(STACS '22), which require $\Omega(\log^3 n)$ rounds of adaptivity. Finally, as a consequence of our edge sampling algorithm, we obtain a $\tilde O(n\log^ 8 n)$ query algorithm for connectivity, using two rounds of adaptivity. This improves on a {three}-round algorithm of Assadi et al.~(ESA '21) and is tight; there is no non-adaptive algorithm for connectivity making $o(n^2)$ queries.
\end{abstract}

\section{Introduction}\label{sec:intro}

In this work, we study sub-linear query algorithms for estimating the number of edges in a simple, unweighted graph $G = (V,E)$, and for sampling uniformly random edges. Access to $G$ is via a \emph{Bipartite Independent Set (BIS)} oracle \cite{beame2020edge}. A query to this oracle takes as input two disjoint subsets $L, R \subseteq V$ and returns
$$\Bq(L, R) = \begin{cases}
\yes & \text{ if there is no edge between $L$ and $R$}\\
\no & \text{ otherwise.}
\end{cases}
$$

\noindent \textbf{Local Query Models.} Prior work on sub-linear query graph algorithms has largely focused on \emph{local} queries, in particular, (i) vertex degree queries (ii) neighbor queries (output the $i$th neighbor of a vertex) and (iii) edge existence queries \cite{feige2006sums, goldreich2008approximating, seshadhri2015simpler}. In the  literature, the first two types of queries form the \textit{adjacency list} query model, while all three types of queries form the \textit{adjacency matrix} query model. Under these  models, a variety of graph  estimation problems have been well studied, including edge counting and sampling~\cite{eden2018sampling, goldreich2008approximating,seshadhri2015simpler,  tvetek2022edge}, subgraph counting~\cite{aliakbarpour2018sublinear, biswas2021towards, eden2020approximating}, vertex cover~\cite{behnezhad2021time, onakRRR12}, and beyond~\cite{ron2019sublinear}.

For a graph with $n$ nodes and $m$ edges, given access only to degree queries, Feige \cite{feige2006sums} presented an algorithm for estimating $m$ up to $(2\pm \epsilon)$ relative error with query complexity  $O(\sqrt{n} \cdot \poly(1/\epsilon,\log n))$. This work also showed that any $(2-o(1))$-approximation algorithm requires $\Omega(n)$ queries. In the  adjacency list query model, Goldreich and Ron~\cite{goldreich2008approximating} gave a $(1\pm \epsilon)$-approximation algorithm, with query complexity $O(n/\sqrt{m} \cdot \poly(1/\epsilon,\log n))$. Recently, Eden and Rosenbaum~\cite{eden2018sampling} gave algorithms for near-uniform  edge sampling with the same query complexity, and showed that this complexity is nearly tight. 

\medskip 

\noindent \textbf{Global Query Models.} Motivated by the desire to obtain more query efficient algorithms, Beame et al.~\cite{beame2020edge} studied edge estimation using \emph{global} queries that can make use of information across the graph, including the BIS queries that we will focus on, and the related Independent Set (IS) queries. IS queries were introduced in the literature on query efficient graph recovery \cite{abasi2019learning, angluin2008learning}. They answer whether or not there exist any edges in the induced subgraph on a subset of nodes $S \subseteq V$. 
We refer the reader to the exposition in~\cite{beame2020edge}, which discusses applications of these global query models in group testing~\cite{chen1990using, dorfman1943detection}, computational geometry~\cite{aronov2008approximating, cabello2015shortest, fishkin2003disk}, fine-grained complexity~\cite{ dell2021fine, dell2020approximately}, and decision versus counting complexity~\cite{dell2020approximately, ron2016power, stockmeyer1983complexity,stockmeyer1985approximation}. 

In the IS query model, \cite{beame2020edge, chen2020nearly} give a $O(\min\{ \sqrt{m}, n/\sqrt{m} \} \cdot \textup{poly}(\log n, 1/\epsilon))$ query algorithm for $(1\pm \epsilon)$ approximate  edge counting. In the  BIS  model, numerous authors~\cite{beame2020edge,dell2020approximately,bhattacharya2022faster} achieve $(1\pm \epsilon)$-approximation  for edge counting and near-uniform edge sampling using just $\poly(1/\epsilon,\log n)$ queries. This is exponentially smaller than the query complexities in the IS and local queries models.

Extending the BIS query model to hypergraphs, Dell et al.~\cite{dell2020approximately} introduce the \textit{coloured independence oracle} which detects the presence of a size $k$ hyperedge. They  give algorithms for hyperedge estimation and sampling using this generalized oracle. 
Many other variants of {global} queries have been studied including  \textsc{Linear, OR} and \textsc{Cut} queries~\cite{assadi2021, chakrabarti2021element, rubinstein18}. These queries have been applied to solving maximum matching~\cite{konrad2020constructing, nisan2021demand}, minimum cut~\cite{rubinstein18}, triangle estimation~\cite{ bhattacharya2019hyperedge, bhattacharya2021triangle, dell2020approximately}, connectivity~\cite{assadi2021}, hitting sets~\cite{bishnu2018parameterized}, weighted edge estimation~\cite{bishnu2019inner}, problems related to linear algebra~\cite{rashtchian2020vector}, quantum algorithms~\cite{montanaro2020quantum}, and full graph recovery~\cite{abasi2019learning, angluin2008learning}.

\medskip

\noindent\textbf{The Role of Adaptivity.} Notably, for both local and global queries, most sub-linear time graph  algorithms are \textit{adaptive}, i.e., a query may depend on the answers to previous queries. In many cases, it is desirable for queries to be \emph{non-adaptive}. This allows them to be completed independently, and might allow for the resulting algorithm to be easily implemented in massively parallel computation  frameworks~\cite{karloff2010model}. Non-adaptive algorithms also lead naturally to single-pass, rather than multi-pass, streaming algorithms. In fact, the \textsc{BIS} query model can be seen as a very restricted subset of the more general \textsc{Linear} query model, in which each query outputs the inner product of the edge indicator vector with a  query vector. This model has long been studied in the graph-streaming literature \cite{ahn2012analyzing,mcgregor2014graph}, in part due to its usefulness in giving single-pass algorithms. However, it has remained open whether non-adaptive algorithms can be given in  more restricted global query models. 

For these reasons, Assadi et al.~\cite{assadi2021} and Chakrabarti and Stoeckl \cite{chakrabarti2021element} have recently sought to reduce query adaptivity under a variety of global query models, including \textsc{Linear, OR, Cut} and \textsc{BIS} queries. 
These works study the \emph{single element recovery} problem, which is a weaker variant of  uniform edge sampling, requiring that the algorithm return a single edge in $G$. Assadi et al.~also study the problem of checking connectivity, presenting a BIS query algorithm making $\tilde O(n)$ queries and using three rounds of adaptivity. They give a two-round algorithm in the stronger OR query model, and show that even  in this model, there is no non-adaptive algorithm for connectivity making $o(n^2)$ queries.

We note that reducing query adaptivity is also a well-studied direction in the closely related literature on group testing~\cite{du2000combinatorial, indyk2010efficiently}. IS and BIS oracles can be thought of as tests if there is a single element in a group of edges, where that group is required to be all edges incident on one node set (IS) or between two disjoint sets (BIS). Attempts to minimize query adaptivity have also been made for sparse recovery~\cite{indyk2011power, kamath2019adaptive, nakos2018improved}, sub-modular function maximization~\cite{balkanski2018adaptive, chekuriQ19}, property testing~\cite{canonne2018adaptivity} and multi-armed bandit learning~\cite{agarwal2017learning}.

\subsection{Our Contributions}

Our main result is the first \emph{non-adaptive} algorithm for edge estimation up to $(1\pm \epsilon)$ relative error, using $\poly(1/\epsilon,\log n)$ BIS queries. Formally, we show:

\begin{theorem}[Theorem~\ref{thm:edge_estimate} restated]\label{thm:intro} 
Given a graph $G$ with $n$ nodes and $m$ edges, there is an algorithm that makes $O(\epsilon^{-5}\log^{5} n\log^6(\log n))$ non-adaptive BIS queries to $G$ and returns an estimate $\hat m$ satisfying:  $m(1-\epsilon) \le \hat m \le m (1+ \epsilon), \textup{with probability at least }3/5$.\footnote{Note that the success probability can be boosted in the standard way, by running multiple independent instantiations of the algorithm and taking their median estimate.}
\end{theorem}

Prior methods for $(1\pm \epsilon)$ error edge estimation using BIS queries are based on a binary search style approach~\cite{beame2020edge,dell2020approximately,bhattacharya2022faster}, which is inherently adaptive, and this leads to algorithms requiring $\Omega(\log^2 n)$ rounds of adaptivity. Beame et al.~\cite{beame2020edge} present a non-adaptive algorithm giving a  $O(\log^2 n)$ approximation factor for bipartite graphs, using $O(\log^3 n)$ queries. However, no non-adaptive results for general graphs or achieving $1\pm \epsilon$ relative error for arbitrary $\epsilon>0$ were previously known.
Even with adaptivity, the best known algorithm due to \cite{bhattacharya2022faster} has a query complexity of $O(\epsilon^{-2}\log^{11} n)$ and succeeds with probability $1-1/n^2$. Therefore, our non-adaptive result improves upon the current best known algorithms, for constant $\epsilon$.

Our second result builds on our edge estimation approach, giving the first  non-adaptive BIS query algorithm that returns a  {near}-uniformly sampled edge. Formally:
\begin{theorem}[Theorem~\ref{thm:sampling} restated]\label{thm:introSampling}
Given a graph $G$ with $n$ nodes, $m$ edges, and edge set $E$, there is an algorithm that makes $O(\epsilon^{-4}\log^6 n\log(\epsilon^{-1}\log n) + \epsilon^{-6} \log^5 n \log^{6} (\log n) \log(\epsilon^{-1} \log n) )$ non-adaptive BIS queries which, with probability at least $1-\epsilon$, outputs an edge from a probability distribution $P$ satisfying   $(1-\epsilon)/{m} \le P(e) \le (1+\epsilon)/{m}$ for every $e \in E$.
\end{theorem}
Prior results for near-uniform edge sampling required $\Omega(\log^3 n)$ rounds of adaptivity~\cite{bhattacharya2022faster, dell2020approximately}. Additionally, even ignoring adaptivity, our results improves on the best known query complexity of $O(\epsilon^{-2} \log^{14} n)$, due to~\cite{bhattacharya2022faster}, for constant $\epsilon$. 

By combining Theorem \ref{thm:introSampling} with prior work on sublinear query graph connectivity, via edge sampling, we obtain a connectivity algorithm using two rounds for adaptivity:
\begin{theorem}[Theorem~\ref{thm:graph_connectivity} restated]\label{thm:introConnectivity}
Given a graph $G$ with $n$ nodes, there is a $2$-round adaptive algorithm that determines if $G$ is connected with probability at least $1-1/n$ using $\Tilde{O}(n\log^{8} n)$ BIS queries,  where $\Tilde{O}(\cdot)$ ignores the $\log^{O(1)} \log n$ dependencies.
\end{theorem}

Theorem \ref{thm:introConnectivity}  improves on a three-round algorithm of Assadi et al.~\cite{assadi2021} and is tight: even in the stronger OR query model (which allows checking the presence of an edge within an arbitrary subset of edges) no non-adaptive algorithm can make $o(n^2)$ queries. Assadi et al.~gave a two-round algorithm in this  stronger OR query model. Thus, Theorem \ref{thm:introConnectivity} closes the gap between BIS queries and OR queries for this problem. We note that there is a separation from the even stronger \textsc{Linear} query model, where non-adaptive algorithms for connectivity and cut approximation are well-known \cite{ahn2012analyzing}. Understanding if there remain interesting separations between the BIS and OR query models in terms of adaptivity would be very interesting.

\section{Technical Overview}
In this section, we present an overview of our non-adaptive BIS query algorithms for edge estimation (Theorem~\ref{thm:intro}) and near-uniform edge sampling (Theorem~\ref{thm:introSampling}), along with our $2$-round algorithm for connectivity (Theorem~\ref{thm:introConnectivity}).

\subsection{Edge Estimation}
A simple idea to estimate the number of edges in a graph via BIS queries is to sample small random subsets of nodes and run BIS queries to check the presence of an edge between them. The fraction of these queries that return \yes (i.e., indicating the presence of no edge) can then be used to estimate the number of edges. In particular, for a graph containing $m$ edges, if the random subgraphs have $O(n/\sqrt{m})$ nodes in them, then we expect a \yes answer with constant probability. 
Beame et al.~\cite{beame2020edge} describe a non-adaptive algorithm along these lines, which gives a $O(\log^2 n)$ approximation for bipartite graphs using $O(\log^3 n)$ queries. Unfortunately, going beyond this coarse approximation factor is  difficult since many dependencies due to common neighbors arise and this increases the variance of the estimators. Beame et al.~handle the issue by using the coarse estimates to subdivide the graph into smaller sub-graphs, until these divided graphs only contain $O(\log^{O(1)} n)$ edges, at which point all their edges can be discovered with few queries. This strategy yields a $(1 \pm \epsilon)$ approximation, however, it is inherently adaptive. 

Our non-adaptive edge estimation algorithm takes a different approach. Suppose we could sample each node with probability  $p_v\approx \epsilon^{-2}d(v)/m$ and compute the degree of the sampled nodes then it is straightforward to show $\sum_v \mathbb{I}[v \mbox{ sampled}]\cdot d(v)/p_v$ equals $2m$ in expectation. Furthermore, an application of Bernstein bound implies that it is a $(1\pm \epsilon)$ with sufficient probability. The challenge is showing that this type of approach can be approximated in the BIS query model.

\medskip

\noindent\textbf{Subsampling Nodes.} The first idea, drawn from work on streaming algorithms, is to subsample the nodes of $G$ at different rates of the form  $1/\gamma^{j}$ where $\gamma>1$ is constant and $j \in \{0, 1, \cdots, O(\log n)\}$. At each rate, we will ``recover'' all sampled nodes (along with a corresponding degree estimate) whose degree is roughly $d(v) \approx {\epsilon^2 m}/{\gamma^{j}}$. In this way, each node will be recovered with probability roughly ${1}/{\gamma^j} \approx\epsilon^{-2}d(v)/m$, as desired. We describe this subsampling procedure in Section~\ref{subsec:estimator}, as part of our main algorithm~\estimator (Algorithm~\ref{alg:estimator}).

\medskip

\noindent\textbf{Recovering Heavy Nodes.} 
The next challenge is to show that we can actually recover the appropriate nodes and degree estimates at each sampling rate. If we can approximate the degree of all nodes sampled at rate $1/\gamma^j$ up to additive error $O\left ({\epsilon^3 \cdot m}/{\gamma^j}\right )$, we will obtain a $(1\pm \epsilon)$ relative error approximation to the degree of any node we hope to recover at that sampling rate, i.e., any node with degree roughly ${\epsilon^2 m}/{\gamma^{j}}$. Using these approximations, we can determine which nodes should be recovered at that rate, and form our edge estimate.

\medskip 

\noindent\textbf{Degree Estimation via Neighborhood Size Estimation.} To achieve such an additive error approximation, we also use ideas from the sparse recovery and streaming literature. In particular, we implement an approach reminiscent of the Count-Min sketch algorithm~\cite{cormode2005improved}. The approach is described in detail in Section~\ref{subsec:degree_counting}, where we present Algorithm~\approxdegree (Algorithm~\ref{alg:count}). First observe that when sampling at rate $1/\gamma^j$, conditioned on any node $v$ being included in the sample, the expected total degree of the sampled nodes other than $v$ is $O(m/\gamma^j)$. If we further subdivide these nodes into $\tilde O(1/\epsilon^3)$ random groups, the expected total degree of all nodes other than $v$ in any group is $\tilde O\left ({\epsilon^3 \cdot m}/{\gamma^j}\right )$. 

Now, if $v$ is placed in group $S$, we can approximately upper bound its degree by the total \emph{neighborhood size of $S$}. This upper bound holds approximately as long as $v$ does not have too many neighbors in $S$, which it won't with good probability. The neighborhood size of $S$ is in turn upper bounded by the degree of $v$ plus the total degree of other nodes in $S$, and thus by $d(v) + \tilde O\left ({\epsilon^3 \cdot m}/{\gamma^j}\right )$ in expectation. So, in expectation, this approach gives an additive $\tilde O\left ({\epsilon^3 \cdot m}/{\gamma^j}\right )$ error approximation to the degree of each sampled node $v$, with constant probability. Repeating this procedure  $O(\log n)$ times, and, as in the Count-Min sketch, taking the minimum degree estimate for each node sampled at rate $1/\gamma^j$,  gives us high probability approximation for such nodes.

\medskip

\noindent\textbf{Neighborhood Size Estimation.} The final step is to implement an algorithm that can estimate the neighborhood size of the random subset of nodes $S$, to be used in our degree estimation procedure. We do this in Section~\ref{subsec:nbr_counting}, where we present Algorithm~\nbralg (Algorithm~\ref{alg:nbrsize}). This algorithm takes as input two disjoint subsets $L, R$ and returns a $(1\pm \epsilon)$-approximation for the size of the neighborhood of $L$ in $R$. We highlight that this may be very different than the \emph{number of edges connecting $L$ to $R$} -- the neighborhood size is the number of nodes in $R$ with at least one edge to $L$. This difference is critical in removing the correlations discussed previously due to common neighbors. Such correlations lead to the adaptive nature of prior algorithms~\cite{beame2020edge, dell2020approximately}. To estimate the size of the neighborhood of $L$ in $R$, we sample the nodes in $R$ at different rates and ask BIS queries on $L$ and the sampled subset of $R$. Intuitively, when the sampling rate is the inverse of the size of the neighborhood, we will observe a \yes response with constant probability. We can detect this and thus estimate the neighborhood size. 
\medskip

\noindent\textbf{Non-adaptivity.}  The approach is inherently non-adaptive as all random sampling of nodes and random subsets can be formed ahead of time, independently of any query responses. The only catch is that to determine which nodes should be recovered at each sampling rate, i.e., those nodes with degree $d(v) \approx \epsilon^{-2} \cdot m/{\gamma^j}$, we need a coarse estimate to the edge count $m$ in the first place. Fortunately, we can bootstrap such an estimate starting with a very coarse $O(\log^2 n)$-relative error approximate estimation, due to Beame et al.~\cite{beame2020edge}. We then {refine} this estimate iteratively using Algorithm~\rEstimate (Algorithm~\ref{alg:refineEst}). Each {refinement} improves the approximation factor by $\epsilon$, and after $O(\log_{1/\epsilon} \log n)$, refinements our estimate will result in a $(1\pm \epsilon)$-approximation factor. The key observation here is that each {refine} step does not require any additional BIS queries. Thus, our algorithm remains non-adaptive.

\subsection{Uniform Edge Sampling and Connectivity}\label{sec:samplingoverview}
In the full version, we prove Theorem \ref{thm:introSampling} by designing and analyzing a {non-adaptive}  algorithm for returning a {near-}uniform sample among the edges of the graph. 
Our approach builds heavily on our edge estimation algorithm. If we knew the degree $d(v)$ of all vertices, then to sample a uniform edge, we could sample a vertex $v \in V$ with probability $d(v)/\sum_{w \in V} d(w)$ and return a uniform neighbor among the neighbors of $v$. We can observe that the probability that an edge $(v, u)$ is sampled is $d(v)/\sum_{w \in V} d(w) \cdot 1/d(v) + d(u)/\sum_{w \in V} d(w) \cdot 1/d(u) = 1/m$, i.e., this approach yields a uniformly random edge sample.

\medskip

\noindent\textbf{Node Sampling.}
We implement the above approach approximately using BIS queries in  Algorithm~\sampling.  First note that recovered vertices in our edge estimation algorithm are sampled with probabilities roughly proportional to their degrees. We argue that we can select a random vertex from this set, which overall is equal to any vertex $v$  with probability approximately $d(v)/\sum_{w \in V} d(w)$. To do so, we leverage our degree estimates, and the fact that our edge count estimator, which is the sum of scaled degrees of recovered vertices, is well-concentrated. 

\medskip

\noindent\textbf{Random Neighbor Sampling.}
It remains to show how to return a random neighbor of the sampled vertex. To do so, in the full version, we describe an algorithm that takes as input two disjoint subsets $L, R$ and returns a uniform neighbor among the neighbors of $L$ in $R$. By showing an equivalence between the substantially more powerful OR queries and BIS queries in this specific setting, we argue that an existing algorithm for OR queries can be extended to return a uniform neighbor using BIS queries. An OR query takes as input a subset of pairs of vertices and returns \yes iff there is an edge in the subset queried. Building on this, in the full version, we present Algorithm~\uneighbor that takes as input the subset of nodes sampled at any rate $1/\gamma^j$ as in our edge estimation algorithm, and approximately returns a uniform neighbor for every vertex $v$ sampled in this set. As before, we construct $\tilde O(1/\epsilon^4)$ random partitions of the sampled nodes. For every vertex $v$ in a random subset $S$, we return a uniform neighbor (obtained using the idea just described) of $S$ as the neighbor of $v$. If $v$ has large degree compared to the total degree of nodes in the partition, which it will if it is meant to be recovered at that sampling rate, this output will most likely be a neighbor of $v$, and will be close to a uniformly random one.

\medskip

\noindent \textbf{A Two-Round Algorithm for Connectivity}. Our non-adaptive edge sampling algorithm (Theorem \ref{thm:introSampling}) yields a two-round algorithm for graph connectivity (Theorem \ref{thm:introConnectivity}), improving on a prior three-round algorithm of~\cite{assadi2021}. In particular, the algorithm of~\cite{assadi2021} selects $O(\log^2 n)$ random neighbors per vertex, and contracts the connected components of this random graph into \emph{supernodes}. This random sampling step can be performed using one round of $\tilde O(n)$ BIS queries. They prove that in the contracted graph on the supernodes, there are at most $O(n\log n)$ edges. Using this fact, they then show how to identify whether all the supernodes are connected using $\tilde O(n)$ BIS queries and two additional rounds of adaptivity.

We follow the same basic approach: using a first round of $\tilde O(n)$  queries to randomly sample $O(\log^2 n)$ neighbors per vertex and contract the graph into supernodes. Once this is done, we observe that we have BIS query access to the contracted graph simply by always grouping together the set of nodes in each supernode. So, we can directly apply the non-adaptive sampling algorithm of Theorem~\ref{thm:introSampling} to sample edges from the contracted graph. By a coupon collecting argument, drawing $O(n\log^2 n)$ near-uniform edge samples (with replacement) from the contracted graph suffices to recover all $O(n \log n)$ edges in the graph, and thus determine connectivity of the contracted graph, and, in turn, the original graph. 
\section{Preliminaries}\label{sec:prelim}
Let  $G(V, E)$  denote the graph  on  vertex set $V$  with edges $E \subseteq V \times V$. Let $|V| = n$ be the number of nodes and $|E| = m$ be the number of edges .  
For any set of nodes $S \subseteq V$, let $E[S] \subseteq E$ denote the edges in the induced subgraph on $S$. For any two disjoint sets of nodes $L, R \subseteq V$, let $E[L, R] = \{ (u, v)\in E \mid u \in L, v \in R \}$ denote the edges between them.
 For any $v \in V$, let $\Gamma(v) = \{ u \mid (v, u) \in E \mbox{ for some $v\in V$} \}$ be its set of neighbours. Let $d(v)=|\Gamma(v)|$ be its degree.
For $S \subseteq V$, let $\Gamma(S) =\bigcup_{u \in S} \Gamma(u)$ and let $d(S) = \sum_{u \in S} d(u)$. 

\begin{definition}[OR query]\label{def:or_query}
An OR query takes as input a collection $E_q$ of pairs of vertices given by $E_q = \{ (x_1, y_1), (x_2, y_2), \cdots (x_k, y_k) \mid x_i, y_i \in V \ \forall i \in [k] \}$ and satisfies the following: 
$$\Or(E_q) = \begin{cases}
\yes \text{ if } E_q \cap E = \phi\\
\no \text{ otherwise.}
\end{cases}
$$
\end{definition}

\begin{lemma}[Bernstein's inequality]\label{lem:bernstein}
Let $X_1, X_2, \ldots X_n$ be independent random variables. Suppose $|X_i| \leq M \ \forall i \in [n].$ Then:
\[ \Pr \left[ \left|\sum_i X_i - \E[X_i] \right| \geq t \right] \leq \exp{\left(-\frac{t^2}{2\sum_i \E[(X_i - \E[X_i])^2] + \frac{2}{3}Mt} \right)}\]
\end{lemma}

\begin{fact}\label{fact:ineq1}
\[ \left( 1+\frac{x}{n}\right)^n \ge e^x \left( 1 - \frac{x^2}{n} \right) \ge e^x \ \textup{ for } |x| \le n, \ n \ge 1.\]
\end{fact}
\section{Non-adaptive algorithm for edge estimation}\label{sec:counting}

In this section, we present our non-adaptive algorithm for edge estimation using BIS queries. In Section~\ref{subsec:nbr_counting}, we describe an algorithm that takes as input two disjoint subsets $L, R$ and returns an estimate of the size of the neighborhood $|\nbr{L} \cap R|$. Next, in Section~\ref{subsec:degree_counting}, we use this algorithm to give additive error approximations of degrees of all the vertices in a given subset. Finally, in Section~\ref{subsec:estimator}, using the approximate degree estimates, we  construct a $(1\pm \epsilon)$-approximate estimator for $m$ by sampling nodes with probabilities roughly proportional to their degrees.

\subsection{Estimating the size of neighborhood}\label{subsec:nbr_counting}

Algorithm \nbralg takes as input two disjoint subsets $L, R \subseteq V$ and returns a $(1\pm \epsilon)$-approximation of the size of neighborhood of $L$ in $R$, i.e., $|\nbr{L} \cap R|$ using $\poly(1/\epsilon,\log n)$ BIS queries. We overview the analysis of this algorithm here, before presenting the details in section~\ref{subsubsec:nbrsize}.

The main idea is to sample subsets of vertices in $R$ (denoted  $\hat R_1,\hat R_2,\ldots$) with exponentially decreasing probability values $1/2,1/4,1/8,\ldots$. When the sampling rate $1/2^i$ falls below ${1}/{|\nbr{L} \cap R|}$, we expect $L$ to no longer have any neighbors in $\hat R_i$ with good probability. In particular, we can return the inverse of the smallest probability $1/2^i$  for which $\bis{L}{\hat R_i} = \yes$, as a coarse estimate for $|\nbr{L} \cap R|$. 

To boost the accuracy of this  estimate, we repeat the  process $T = O(\epsilon^{-2}\log(\delta^{-1} \cdot \log n))$ times, and at each sampling rate count the number of times the BIS query $\bis{L}{\hat R_i}$ returns $\yes$. This count is denoted $\Count{i}$ in Algorithm~\ref{alg:nbrsize}, and its expectation can be written in closed form as $\E[\Count{i}] = T \cdot (1-1/2^i)^{|\nbr{L} \cap R|}$. Suppose $2^{\hat i}\le |\nbr{L} \cap R |< 2^{\hat i+1}$, then, $\E[\Count{\hat i}] =\Theta(T)$. Via a standard Chernoff bound, it will be approximated to $(1 \pm \epsilon)$ error with high probability by $\Count{\hat i}$. Thus, we can compute an accurate estimate of the neighborhood size by inverting our estimate of $\E[\Count{\hat i}]$, as $\log_{(1-{1}/{2^{\hat i} })} ( {\Count{\hat i}}/{T} )$. We identify the appropriate $\hat i$ in line 12 of Algorithm~\ref{alg:nbrsize}, and compute the corresponding estimate in lines 13-14. There is one edge case handled in line 13: if $|\nbr{L} \cap R | = 1$ we will have $\hat i = 0$, and $\Count{\hat i} = 0$. The final error bound for Algorithm~\ref{alg:nbrsize} is stated below.

\begin{algorithm}[!ht]
\caption{\nbralg: Estimating the neighborhood size of $L$ in $R$}
\label{alg:nbrsize}
\begin{small}
\begin{algorithmic}[1]
\Statex \textbf{Input:} $L , R \subseteq V$,  approximation error  $\epsilon$, failure probability $\delta$.
\Statex \textbf{Output:} $\nbrest{}(L)$ as an estimate of $|\nbr{L} \cap R|$.
\State Initialize $\nbrest{}(L) \leftarrow 0$.
\For{$i = 0,1,\ldots \log_2 n$}
\State $\Count{i} \leftarrow 0$. 
\For{$t = 1, 2, \ldots T =  {2e^8 \ln(\log n/\delta)}\cdot \epsilon^{-2} $}
\State $\Hat{R}_i^t \leftarrow \{ u \in R \mid u \text{ is included independently with probability }1/2^i \}$.
\State $\Count{i} = \Count{i} + \bis{L}{\Hat{R}^t_i}$
\EndFor 
\EndFor
\If{$\Count{0} = T$}
\State \Return $\nbrest{}(L) = 0$.
\Else
\State Set $\hat i \leftarrow \max \left\{ i \mid \frac{\Count{i}}{T} < \frac{(1-\epsilon)}{2e^2} \right\}$.
\If{$\hat i = 0$} \Return $\nbrest{}(L) = 1$.
\Else
 \, \Return $\nbrest{}(L) = \log_{(1-{1}/{2^{\hat i} })} ( {\Count{\hat i}}/{T} )$.
\EndIf
\EndIf
\end{algorithmic}
\end{small}
\end{algorithm}

\subsubsection{Approximation Guarantees of Algorithm~\nbralg}\label{subsubsec:nbrsize}

For any $i \in \{0, 1,\ldots, \log_2 n\}$ let $\Hat{R}_{i}$ denote a set constructed by sampling vertices of $R$ with probability $1/2^{i}$. In Algorithm~\ref{alg:nbrsize}, we construct $T = O(\epsilon^{-2} \log(\delta^{-1} \cdot \log n))$ such  sets, each denoted by $\Hat{R}^t_{i} \ \forall \ t \in [T]$. Let $\Count{i} = \sum_{t = 1}^T \bis{L}{\Hat{R}^t_i}$ denotes the number of times the BIS query $\bis{L}{\Hat{R}^t_i}$ returns \yes. For any $t \in [T]$, we define: \[ p(i) =  \Pr \left[ \bis{L}{\Hat{R}^t_i} = \yes \right] = \Pr \left[ \nbr{L} \cap \hat R^t_{i} = \phi \right] \textup{ and } \hat p(i) =  \frac{\Count{i}}{T}.\] 

Suppose $L$ satisfies:
\[ 2^{i^*} \leq |\nbr{L} \cap R| < 2^{i^*+1} \textup{ for some }i^* \in \{0, 1,\ldots, \log_2 n\}. \]

\begin{claim}
\label{cl:lb}
We have the following bounds:
\[ p(i^*) \ge \frac{1}{2e^2}, \quad p(i^*-2) > \frac{1}{2e^8}, \textup{ and } p(i^*-2) \le \frac{1}{e^4}.\]
\end{claim}
\begin{proof}

\begin{align*}
    p(i) = \Pr \left[ \nbr{L} \cap \hat R^t_{i} = \phi \right] = \Pr[u \not\in \Hat{R}^t_{i} \ \forall u \in R \cap \nbr{L}] &= \prod_{u \in R \cap \nbr{L}} \Pr[u \not\in \Hat{R}^t_{i}] &= \left( 1- \frac{1}{2^{i}} \right)^{|\nbr{L} \cap R|}.
\end{align*}

We can lower bound $p(i^*)$ by
\begin{align*}
    p(i^*) = \left( 1- \frac{1}{2^{i^*}} \right)^{|\nbr{L} \cap R|} \geq \left( 1- \frac{1}{2^{i^*}} \right)^{2^{i^*+1}} &\ge e^{-2}\left( 1-\frac{2^{i^*+1}}{2^{2i^*}} \right) \textup{ (using inequality}~\ref{fact:ineq1})\\ &\ge \frac{1}{2e^2} \quad \textup{ for }i^*  
    \ge 2.\\
    p(i^*) &\ge \frac{1}{8} \quad \textup{ for } i^* = 1.
\end{align*}

If $i = i^*-2$, we have:
\[
    p(i) = \left( 1- \frac{1}{2^{i}} \right)^{|\nbr{L} \cap R|} \le \left( 1- \frac{1}{2^{i}} \right)^{2^{i^*}} \le e^{-2^{i^*-i}} = \frac{1}{e^4}\]
    \begin{align*}
            p(i) =  \left( 1- \frac{1}{2^{i}} \right)^{|\nbr{L} \cap R|} \ge \left( 1- \frac{1}{2^{i^*-2}} \right)^{|\nbr{L} \cap R|} &> \left( 1- \frac{1}{2^{i^*-2}} \right)^{2^{i^*+1}} \\
            &\ge e^{-8}\left( 1- \frac{2^{i^*+1}}{2^{2i^*-4}} \right) \textup{ (using inequality}~\ref{fact:ineq1})\\ 
            &\ge \frac{1}{2e^8} \textup{ for } i^* \ge 5.
    \end{align*}
    
For $i^* \le 4$, the inequality is satisfied. So, we have:
$p(i^*-2) > 1/2e^8$.

\end{proof}

\begin{claim}\label{cl:nbr_concentration}
For sufficiently small $\epsilon$, with probability at least $1-\delta$, we have:
\[ \Count{i} \ge \frac{1-\epsilon}{2e^2} \cdot T \ \forall i \ge i^* \mbox{ and }\Count{i^*-2} < \frac{1-\epsilon}{2e^2} \cdot T.\]
\end{claim}
\begin{proof}
As $\Count{i} = \sum_{t = 1}^T \bis{L}{\Hat{R}^t_i}$, we have: $\E[\Count{i}] = T \cdot p(i)$. Using Claim~\ref{cl:lb}, we have:
 \[T = {2e^8 \ln(\log n/\delta)}\cdot \epsilon^{-2} \ge \frac{4\ln(\log n/\delta) \cdot \epsilon^{-2} }{p(i^*-2)} \ge \frac{4\ln(\log n/\delta) \cdot \epsilon^{-2} }{p(i^*)}, \mbox{ as $p(i^*) \ge p(i^*-2)$.}\]
 
Suppose $i \in \{i^*, i^*-2\}$. Then, we have:
\begin{align*}
    \Pr \left[ |\hat p(i) - p(i)| \ge \epsilon \cdot p(i) \right] &=   \Pr \left[ |T \cdot \hat p(i) - T \cdot p(i)| \ge T \cdot \epsilon \cdot p(i) \right]\\
    &= \Pr \left[ |\Count{i} - \E[\Count{i}]| \ge \epsilon \E[\Count{i}] \right]\\
    &\le 2\exp{\left( - \frac{\epsilon^2 T p(i) }{2}\right)} \le \frac{\delta}{\log n} \quad (\textup{Using Chernoff bound}).
\end{align*}

Using Claim~\ref{cl:lb}, we get:
\begin{align*}
   \Count{i^*} &\ge (1-\epsilon) \cdot T \cdot p(i^*) \ge \frac{(1-\epsilon)}{2e^2} \cdot T \\
   \Count{i^*-2} &< (1+\epsilon) \cdot T \cdot p(i) \le \frac{(1+\epsilon)}{e^4} \cdot T\\
   \Longrightarrow \Count{i^*-2} &< \frac{(1+\epsilon)}{e^4} \cdot T \le \frac{(1-\epsilon)}{2e^2} \cdot T, \textup{ when } \epsilon \le \frac{e^2/2-1}{e^2/2+1}.
\end{align*}

 From the definition, we can observe that $p(i) \ge p(i^*) \ \forall i \ge i^*$. So, the concentration around expected values for $\Count{i}$ obtained using Chernoff bounds will hold for all $i \ge i^*$. Using union bound on at most $\log n$ sampling levels, we have,  with probability $1-\delta$:
 \[\Count{i^*-2} < \frac{(1-\epsilon)}{2e^2} \cdot T \mbox{ and }\Count{i} \ge \frac{(1-\epsilon)}{2e^2} \cdot T \ \forall i \ge i^*.\]
\end{proof}

\begin{lemma}
\label{lem:nbr}
Algorithm~\ref{alg:nbrsize} uses $O({\epsilon^{-2}}{\log n\log(\delta^{-1} \cdot \log n)})$ BIS queries and returns an estimate $\nbrest{}(L)$ of $|\nbr{L} \cap R|$ such that with probability at least $1-\delta$,
\[ {(1-\epsilon)} \cdot {|\nbr{L} \cap R|} \leq \nbrest{}(L) \leq (1+\epsilon) \cdot |\nbr{L} \cap R| \ .\] 
\end{lemma}
\begin{proof}
If $|\nbr{L} \cap R| = 0$, then, $\Count{i} = T$ for every $i \in \{0, 1, 2, \cdots, \log n\}$. So, $\hat i = 0$, as none of the values $\Count{i}$, for any $i$ will be below the threshold value of $(1-\epsilon) T / 2e^2$. So, the estimate $\nbrest{}(L) = 0$ returned is exact.

Suppose $|\nbr{L} \cap R| = 1$. When we sample with probability $1/2^i$ when $i = 0$, we obtain $\hat R^t_i = R$, for every $t \in [T]$. As $\bis{L}{R} = \yes$, we have $\Count{i} = 0$, and our estimate $\nbrest{}(L) = 1$ is exact. For the remainder of the proof, we assume $i^* \ge 1$.

\medskip

From Algorithm~\ref{alg:nbrsize}, we define $\hat i = \arg \max \{ i \mid \Count{i} < (1-\epsilon) T /2e^2 \} + 1$. From Claim~\ref{cl:nbr_concentration}, this implies: $\hat i \ge i^*-2$. Therefore, with probability at least $1-\delta$, we have: $$ i^* - 2 \le \hat i \le i^*-1. $$
 
Now, we argue that $\nbrest{}(L)$ defined by:
\[ \nbrest{}{(L)} := {\log_{\left(1-{1}/{2^{\hat i} }\right)} \hat p(\hat i)} \ \mbox{obtains a $(1\pm \epsilon)$ approximation for $|\nbr{L} \cap R|$.}\]
\begin{align*}
    (1-\epsilon) p(\hat i) &\le \hat p (\hat i) \le (1+\epsilon) p(\hat i)\\
    \log_{\left(1-{1}/{2^{\hat i} }\right)} (1-\epsilon) \cdot p(\hat i) &\le \log_{\left(1-{1}/{2^{\hat i} }\right)} \hat p (\hat i) \le \log_{\left(1-{1}/{2^{\hat i} }\right)} (1+\epsilon) \cdot p(\hat i)\\
    |\nbr{L} \cap R| + \log_{\left(1-{1}/{2^{\hat i} }\right)} (1-\epsilon) &\le \log_{\left(1-{1}/{2^{\hat i} }\right)} \hat p (\hat i) \le |\nbr{L} \cap R| + \log_{\left(1-{1}/{2^{\hat i} }\right)} (1+\epsilon)\\
    \Rightarrow  |\nbr{L} \cap R| - 2^{ \hat i} \cdot \epsilon &\le \log_{\left(1-{1}/{2^{\hat i} }\right)} \hat p (\hat i) \le |\nbr{L} \cap R| + 2^{\hat i} \cdot \epsilon\\
    \Rightarrow  |\nbr{L} \cap R| - 2^{i^*-2 } \cdot \epsilon &\le \log_{\left(1-{1}/{2^{\hat i} }\right)} \hat p (\hat i) \le |\nbr{L} \cap R| + 2^{i^*-1} \cdot \epsilon\\
    \Rightarrow  (1-\epsilon/4) \cdot |\nbr{L} \cap R| &\le \log_{\left(1-{1}/{2^{\hat i} }\right)} \hat p (\hat i) \le (1+\epsilon/2) \cdot |\nbr{L} \cap R|.
\end{align*}

Therefore, $\nbrest{}(L) := \log_{\left(1-{1}/{2^{\hat i} }\right)} \hat p (\hat i)$ is a $(1\pm \epsilon)$-relative error approximation of $|\nbr{L} \cap R|$.

The total number of BIS queries used by Algorithm~\ref{alg:nbrsize} is $O(\log n \cdot T) = O(\epsilon^{-2}\log n\log(\log n/\delta))$.

\end{proof}
\subsection{Finding good approximation for degrees of vertices}\label{subsec:degree_counting}

We now describe how to use the \nbralg algorithm to estimate the degrees of all vertices in a given subset $S \subseteq V$ up to additive error depending on the total degree of $S$.  Our approach is inspired by the count-min sketch algorithm~\cite{cormode2005improved}. We randomly partition $S$ into subsets $S^1, S^2, \ldots, S^\lambda$ where $\lambda = O(\epsilon^{-3} \log^2 n)$. The choice of the parameter $\lambda$ is based on the analysis in Section~\ref{subsec:estimator}. For each $S^{i}$, we estimate the size of the neighborhood of $S^{i}$ in $V \setminus S^{i}$ using \nbralg. We then return this neighborhood size estimate as the degree estimate for all vertices in $S^{i}$. 
For $v \in S^i$, $|\Gamma(S^i) \cap V \setminus S^i|$ is nearly an overestimate for $d(v)$, as long as $v$ has few neighbors in $S^i$, which it will with high probability. Additionally, it is not too large an overestimate -- we can observe that $|\Gamma(S^i) \cap V \setminus S^i| - d(v) \le d(S^i \setminus v)$. I.e., the error in the overestimate is at most the total degree of the other nodes in $S^i$. In expectation, this error is at most $\frac{d(S)}{\lambda} = O \left (d(S) \cdot \frac{\epsilon^3}{\log^2 n}\right )$ due to our random choice of $S^i$. 
 
As in the count-min sketch algorithm, to obtain high probability estimates, we repeat the process $T = O(\log n)$ times and assign the  minimum among the neighborhood estimates as the degree estimate of $d(v)$. The full approach is given in Algorithm \ref{alg:count} (\textsc{Estimate-Degree}) and the error bound in the Lemma \ref{lem:approxDeg} below. 

\begin{lemma}\label{lem:approxDeg}
Suppose $S \subseteq V$. Then, Algorithm~\ref{alg:count} uses $O(\epsilon^{-5}\log^3 n\log^2(\log n))$ BIS queries and with probability $1-O(1/\log n)$, returns degree estimates $\Hat{d}(v)$ for every vertex $v \in S$ satisfying:
$$ d(v)(1-\epsilon) \leq \Hat{d}(v) \leq d(v) + \frac{\epsilon^3 }{\log^2 n} \cdot d(S).$$
\end{lemma}

\begin{algorithm}[!ht]
\caption{\textsc{Estimate-Degree}: Obtain additive approximate degree estimates}
\label{alg:count}
\begin{small}

\begin{algorithmic}[1]
\Statex \textbf{Input:} $S$ is a subset of $V$, $\epsilon$ is approximation error. 
\Statex \textbf{Output:} Degree estimates of vertices in $S$.
\State Scale $\epsilon \leftarrow \epsilon/3$ and initialize $\Hat{d}(v) \leftarrow n $ for every $v \in S$.
\For{$t$ in $\{1, 2, \ldots, O(\log n)\}$}
\State Consider a random partitioning of $S$ into $S^{t1}, S^{t2}, \ldots S^{t\lambda}$ where $\lambda = O(\epsilon^{-3} \log^2 n)$.
\For{every partition $S^{ta}$ where $a \in [\lambda]$}
\State $\nbrest{}(S^{ta}) \leftarrow \nbralg(S^{ta}, V \setminus S^{ta}, \eps, \delta)$, where $\delta = O(1/\log^4 n)$.
\State $\Hat{d}(v) \leftarrow \min\{\Hat{d}(v),\nbrest{}(S^{ta}) \} \ \forall v \in S^{ta}$.
\EndFor
\EndFor
\State \Return $\Hat{d}(v)$ for every $v \in S$.
\end{algorithmic}
\end{small}
\end{algorithm}

\subsubsection{Proof of Lemma~\ref{lem:approxDeg}}

Consider a vertex $v \in S$. It is easy to observe that in any partition $S^{ta}$ containing $v$, where $t \in [T] \textup{ and } a \in [\lambda]$, the degree of $v$ outside the partition (denoted by $d(v, V\setminus S^{ta})$) is upper bounded by the total size of the neighborhood of $S^{ta}$ (denoted by $|\nbr{S^{ta}} \cap V \setminus S^{ta}|$) which is upper bounded by the total degree of vertices present in the partition (denoted by $d(S^{ta})$). Similar to the analysis of count-min sketch, a simple, yet important observation is that the total degree of the partition except for vertex $v$, i.e., $d(S^{ta} \setminus \{v\})$ is less than $c \cdot d(S)/\lambda$ for some constant $c$ and results in the additive approximation factor of $O(d(S)/\lambda)$. Now, we present the proof of Lemma~\ref{lem:approxDeg}:

\begin{proof}
Given $S \subseteq V$. Consider a vertex $v \in S^{ta}$ for some $a \in \{1, 2 \ldots, \lambda \}$ and $t \in \{1, 2, \cdots, T\}$. For each call to neighborhood size estimation, we set the failure probability to be $\delta = O(1/\log^4 n)$. From Lemma~\ref{lem:nbr}, we have with probability $1-{\delta}$:

\begin{equation*}
    (1-\epsilon)|\nbr{S^{ta}} \cap (V \setminus S^{ta})| \leq \nbrest{}(S^{ta}) \leq (1+\epsilon) |\nbr{S^{ta}} \cap (V \setminus S^{ta})|
\end{equation*}

We can observe that $d(v, V\setminus S^{ta}) = |\nbr{v} \cap (V \setminus S^{ta})| \leq |\nbr{S^{ta}} \cap (V \setminus S^{ta})|$. Therefore:
$$d(v, V\setminus S^{ta}) \leq \frac{\nbrest{}(S^{ta})}{1-\epsilon}.$$

Consider the following:
\begin{align*}
    \E [ d(v, S^{ta}) ] = \E \left[ \sum_{u \in V} \mathbbm{1} \{ u \in \nbr{v} \cap S^{ta} \} \right] 
    = \frac{d(v)}{\lambda} = \frac{d(v)\epsilon^3}{c \log^2 n}
    \leq \epsilon d(v), \mbox{ as $c > 1$}.
\end{align*}

From Markov's inequality, with probability at least $1/2$, we have:  $d(v, S^{ta}) \leq 2\epsilon d(v).$

Combining the above, with probability $1/2-\delta$, we have:
$$
    \frac{\nbrest{}(S^{ta})}{1-\epsilon} \geq d(v, V\setminus S^{ta}) = d(v) - d(v, S^{ta}) \geq d(v) (1-2\epsilon),$$
 $$ \Longrightarrow \nbrest{}(S^{ta}) \ge (1-3\epsilon) d(v).$$

\begin{align*}
    \E\left[d(S^{ta} \setminus \{ v \}) \right] = \E\left[\sum_{u \in S \setminus \{ v \}} d(u) \mathbbm{1}\{ u \in S^{ta} \}\right] &= \sum_{u \in S \setminus \{v\}} d(u) \Pr[u \in S^{ta}]\\
    &= \sum_{u \in S \setminus \{ v \}} d(u) \cdot \frac{1}{\lambda}\\
    &\le d(S) \cdot \frac{\epsilon^3}{c \log^2 n}, \mbox{ for some constant $c > 1$}\\
    &\le d(S) \cdot \frac{\epsilon^3}{\log^2 n}.
\end{align*}

From Markov's inequality, it follows that:
\begin{align*}
    & \Pr\left[d(S^{ta} \setminus \{v\}) \geq  d(S) \cdot \frac{2 \epsilon^3}{\log^2 n} \right] \leq \frac{\E[d(S^{ta} \setminus \{ v \})]}{  d(S) \cdot \frac{2\epsilon^3}{\log^2 n}} = \frac{1}{2}.
\end{align*}

So, with probability at least ${1}/{2}$, we have:
\begin{align*}
        d(S^{ta}) &= d(v) + d(S^{ta} \setminus \{ v\}) \\
             &\leq d(v) +{  d(S) \cdot \frac{2\epsilon^3}{\log^2 n}}\\
\Longrightarrow \nbrest{}(S^{ta}) &\leq |\nbr{S^{ta}} \cap (V \setminus S^{ta})| \leq d(S^{ta}, V \setminus S^{ta}) \leq d(S^{ta})\\
\nbrest{}(S^{ta}) &\leq d(v) +  {  d(S) \cdot \frac{2\epsilon^3}{\log^2 n}}.
\end{align*}

Using union bound on all possible sets $S^{ta}$ for all $t \in [T]$ and $a \in [\lambda]$, with probability at least $1-T \cdot \lambda \cdot \delta \ge 1- 1/2 \log n$, the neighborhood estimates are $(1\pm \epsilon)$-relative approximations. By taking minimum of all the $T = O(\log n)$ estimates, we argue that $\hat d(v)$ is a good approximation of $d(v)$. We take minimum of all the $T$ estimates containing $v$ and obtain the final degree estimate, given by: $$\Hat{d}(v) = \min_{t \in T} \nbrest{}(S^{ta}).$$

We observe that:
\begin{align*}
    \Pr\left[\Hat{d}(v) < (1-3\epsilon) d(v)\right] &= \Pr\left[ \left\{\min_{t \in T} \nbrest{}(S^{ta})\right\} < (1-3\epsilon) d(v) \right]\\
    &= \prod_{t \in T} \Pr\left[ \nbrest{}(S^{ta}) < (1-3\epsilon) d(v)\right]\\
    &\leq \left(\frac{1}{2}\right)^T \le \frac{1}{2n^4}, \mbox{and }\\
    \Pr\left[\Hat{d}(v) > d(v) + {  d(S) \cdot \frac{2\epsilon^3}{\log^2 n}} \right] &= \Pr\left[ \left\{\min_{t \in T} \nbrest{}(S^{ta})\right\} > d(v) + {  d(S) \cdot \frac{2\epsilon^3}{\log^2 n}} \right]\\
    &= \prod_{t \in T} \Pr\left[ \nbrest{}(S^{ta}) > d(v) + {  d(S) \cdot \frac{2\epsilon^3}{\log^2 n}} \right]\\
    &\leq \left(\frac{1}{2}\right)^T \le \frac{1}{2n^4}.
\end{align*}

By taking a union bound on all the vertices in $S$ and the event that neighborhood estimates are accurate, the total failure probability is at most $1/2\log n + 1/n^3 \le 1/\log n$.
Therefore, for every vertex $v \in S^{ta}$, we have with probability at least $1-1/\log n$:
$$(1-3\epsilon)d(v) \leq \nbrest{}(S^{ta}) \leq d(v) + {  d(S) \cdot \frac{2\epsilon^3}{\log^2 n}}.$$

Set $\epsilon = \epsilon/2^{1/3}$ to appropriately scale the value of $\epsilon$ for the final guarantees. Algorithm~\ref{alg:count} uses $O(\epsilon^{-3}\log^2 n \cdot T)$ many calls to the sub-routine $\nbralg$, i.e., Algorithm~\ref{alg:nbrsize}. From Lemma~\ref{lem:nbr}, we know that Algorithm~\ref{alg:nbrsize} uses $O(\epsilon^{-2} \log n\log(\delta^{-1} \log n))$ BIS queries, where we set $\delta = O(1/\log^4 n)$. Therefore, the query complexity of Algorithm~\ref{alg:count} is $O(\epsilon^{-5}\log^4 n\log(\log n))$. 

\end{proof}
\subsection{Edge Estimation}\label{subsec:estimator}

In this section, we describe the algorithm~\estimator that obtains a $(1\pm \epsilon)$-approximation for the number of edges $m$. The constants used $c_1, c_2$ satisfy $c_1 \le c_2/10$ and $c_2 \ge 50$, and we do not explicitly mention them for the sake of brevity. Missing details are presented in the full version.

\medskip

\noindent \textbf{Our Approach}. A naive strategy to estimate the number of edges (denoted by $m$) is to sample roughly $\Tilde{O}(\epsilon^{-2})$ nodes uniformly, and estimate $m$ given the degrees of the sampled nodes. However, the variance of such an estimator depends on the maximum degree, which could be as high as $O(n)$. To fix this issue, we sample vertices at different rates. Our sampling rates are given by the sequence $1/\gamma^{j}$ where $\gamma > 1\textup{ is a constant, } j \in \{0, 1, \cdots, \log n\}$. We use the term $j^{th}$ level to refer to the sampling rate $\gamma^{-j}$. It is easy to observe that when a vertex $v$ is sampled at rate $\Tilde{O}({\epsilon^{-2}d(v)}/{m})$, its contribution is $\Tilde{O}(\epsilon^2 m)$. In other words, if $d(v) \approx {\epsilon^2 m}/{\gamma^j}$, for some sampling level $j$, we can use it in our estimator. However, there are three main challenges in implementing this approach which we detail below.

\medskip

\noindent \textbf{Approximate degrees}. Algorithm~\approxdegree returns degree estimates that are \emph{approximate} with an additive approximation error of $\Tilde{O}\left({\epsilon^3 m}/{\gamma^j}\right)$ at sampling level $j$. To include a vertex $v$, we have to ensure that this error term is small and given by $O(\epsilon d(v))$. When $d(v) = \Tilde{\Omega}({\epsilon^2 m}/{\gamma^j})$, the returned degree estimate $\hat d(v)$ will be a $(1\pm \epsilon)$-approximation to the actual degree $d(v)$. Observe that this corresponds to the threshold we mentioned earlier. Therefore, our goal is to identify all vertices at every level $j$ that pass the threshold of $\Tilde{\Omega}({\epsilon^2 m}/{\gamma^j})$. When that happens, we say that the vertex $v$ has been recovered at level $j$ and can be safely included in our estimator.

\medskip

\noindent \textbf{Knowledge of $m$}. As we do not know the value of $m$, we start with an $O(\log^2 n)$-relative error approximate estimate, obtained by the Algorithm~\textsc{CoarseEstimator} in Beame et al.~\cite{beame2020edge}. We repeatedly refine the approximate estimate using Algorithm~\rEstimate, until we get a $(1\pm \epsilon)$-relative error approximation of $m$. Each \emph{refinement} improves the approximation factor from the previous stage by a multiplicative factor of $\epsilon$. We note that each refinement does not require any additional BIS queries and uses the approximate degree estimates obtained previously.

\medskip

\noindent \textbf{Boundary Vertices}. It is possible that some vertices have degrees close to the threshold values at each sampling level. We denote such vertices $\Vboundary$ and refer to them as  \emph{boundary vertices}. For such boundary vertices, as we use approximate degree estimates, they might be recovered at a level different from its true level (defined with respect to exact degrees). Such a scenario could potentially affect the contribution of the recovered vertex in our estimator by an additional multiplicative factor dependent on $\gamma$ and the difference between recovered level and true level. As a result, our estimator might not be a $(1\pm \epsilon)$-relative error approximation anymore. We get around this limitation by dividing the region between any two consecutive levels into $B$ buckets and shifting the boundaries of all the levels by a random shift selected uniformly from the first $B$ buckets. We account for this by changing the sampling rates to $\gamma^{-\boundary{j}}$ where $\boundary{j}$ encodes the random shift. 

With the random shift of the level boundaries, we ensure that every vertex will lie close to the boundary with probability at most $\epsilon$. Moreover, we argue that every boundary vertex is recovered at its true level or level adjacent to its true level. Therefore, the total contribution of $\Vboundary$ to our edge estimator is $O(\epsilon m)$.

\subsubsection{{Overview of Algorithm }\textsc{Edge-Estimator}}

\noindent \textbf{Random Boundary Shift}. Let $\epsilon$ denote the approximation parameter, $B = 2/\epsilon$ denote the total number of buckets between two consecutive levels and $\gamma = 1/(1-\epsilon)$ the probability of sampling parameter. The region between two consecutive levels is divided into $B$ buckets with the boundaries of buckets proportional to the values given by $\{ [1/\gamma^{B}, 1/\gamma^{B-1}),  \cdots, [1/\gamma^2, 1/\gamma), [1/\gamma, 1) \}$. We select a random integer offset for shifting our levels, denoted by $s$, which is selected uniformly at random from $[0, B)$. Now, the level boundaries are located at values proportional to $\gamma^{-\boundary{j}}$ where $\boundary{j} = j \cdot B - s$ and $0 \leq j \leq L$. Observe that the number of sampling levels is given by $L = \frac{1}{B} \cdot \log_\gamma n + 1\le \frac{1}{2} \log n + 1$. 

\begin{algorithm}[!ht]
\caption{\estimator: Non-adaptive algorithm for estimating edges}
\label{alg:estimator}
\begin{algorithmic}[1]
\Statex \textbf{Input:} $V$ set of $n$ vertices and $\epsilon > 0$ error parameter.
\Statex \textbf{Output:} Estimate $\hat m$ of number of edges in $G$.
\State Scale $\epsilon \leftarrow \frac{\epsilon}{600\log_{1/\epsilon} \log n}$ and initialize $\gamma \leftarrow {1}/{(1-\epsilon)}$ and $B \leftarrow {2}/{\epsilon}$. 
\State Let $s$ be an integer selected uniformly at random from the interval $[0, B)$.
\State Let $\boundary{j} \leftarrow -s + j \cdot B$ for every integer $j$ in the interval $\left[0, \frac{1}{B} \cdot \log_{\gamma} n + 1\right].$
\State Initialize $S_0 \leftarrow V$ and construct $S_1$ by sampling vertices in $S_0$ with probability $1/\gamma^{\boundary{1}}$.
\State Construct $S_2 \supseteq \ldots \supseteq S_L$ for $L = \frac{1}{B} \cdot \log_\gamma n$ where each $S_j$ is obtained by sampling vertices in $S_{j-1} \ \forall j\ge 2$, independently with probability $1/\gamma^B$. 
\For{$j = 0,1,\ldots L$}
\State Run \approxdegree($S_j$) to obtain the estimates $\hat d_j(v)$ for all $v \in S_j$ satisfying: $$(1-\eps) d(v) \le \hat d_j(v) \le d(v) + \frac{c_1\epsilon^3 \cdot m}{\log n \cdot \gamma^{\boundary{j}}}.$$
\EndFor
\State Let $\bar m_0$ be the $O(\log n)$-approximate estimate from the Algorithm \textsc{CoarseEstimator} in Beame et al.~\cite{beame2020edge} on a random partition of $V$.
\State Set $\bar m_0 \leftarrow \max\{2, 16 \log n \cdot \bar m_0\}$, so that we have $m \le \bar m_0 \le (64\log^2 n) \cdot m.$
\For{$t= 1, 2, \cdots, T = 2\log_{1/\epsilon} \log n$}
\State $\bar m_{t}$ is assigned the output of $\rEstimate$ that takes as input approximate degree values $\hat d_j(v) \ \forall v \in S_j \ \forall j \in [L]$, the previous estimate $\bar m_{t-1}$ and the iteration $t$.
\EndFor
\State \Return $\hat m \leftarrow \bar m_{T}$.
\end{algorithmic}
\end{algorithm}

\medskip

In Algorithm~\estimator, we construct sets $V = S_0 \supseteq S_1 \supseteq \cdots \supseteq S_L$ where a set $S_j$ (for all $j \ge 2$) is obtained by sampling vertices in $S_{j-1}$ with probability $1/\gamma^B$. The set $S_1$ is obtained by sampling vertices in $V$ with probability $1/\gamma^{-s+B}$. Our sampling scheme results in each vertex being included in a set $S_j$ with probability $1/\gamma^{\boundary{j}}$. We can easily show that with constant probability, $d(S_j) = O({m \log n}/{\gamma^{\boundary{j}}})$, for all $j$. Using Algorithm~\ref{alg:count}, we obtain approximate degree estimates of vertices in $S_j$ for every sampling level $j \le L$ with an approximation error of $O\left({\epsilon^3}/\log^2 n \cdot d(S_j)\right) = O\left({m \epsilon^3}/{\gamma^{\boundary{j}} \log n}\right)$. By starting with a bad estimate $\bar m_0$ for the total number of edges $m$ and initialized to a $O(\log^2 n)$-approximate estimate, we refine it to obtain an improved estimate $\bar m_1$. We repeat this process $T = 2\log_{1/\epsilon} \log n$ times, such that the estimate $\bar m_{t-1}$ is used to construct an improved estimate $\bar m_t$. Finally, we return the estimate $\bar m_T$  as our final estimate for $m$.

\begin{algorithm}[!ht]
\caption{\textsc{Refine-Estimate}: Refines the current estimate of number of edges}
\label{alg:refineEst}
\begin{algorithmic}[1]
\Statex \textbf{Input:} $\bar m$ satisfying $m \le \bar m \le m(1+\alpha)$, approximate degree values $\hat d_j(v) \ \forall v \in S_j \ \forall j \in [L]$ obtained using Algorithm~\ref{alg:estimator}, $\bar m_0$, and iteration $t$.
\Statex \textbf{Output:} Estimate $\hat m$ satisfying $m \le \hat m \le m(1+\epsilon \cdot \alpha)$ of number of edges in $G$.
\State Initialize $\hat m \leftarrow 0$.
\State Initialize $r(v) \leftarrow 0$ for all $v$ (indicator if $v$ has been recovered yet).
\For{$j = 0,1,\ldots L$}
\For{$v \in S_j$}
\If{$r(v) = 0$ and $\hat d_j(v) \ge \frac{\bar m}{\gamma^{\boundary{j}}} \cdot \frac{c_2\epsilon^2}{\log n}$ }
\State $\hat m \leftarrow \hat m + {\gamma^{\mu(j)}} \cdot \hat d(v)$
\State $\hat \ell(v) \leftarrow j$ and $r(v) \leftarrow 1$.
\EndIf
\EndFor
\EndFor
\If{$t < T = 2\log_{1/\eps} \log n$}
\State $\hat m = {\hat m}/2 + \left(\epsilon \log \log n\right)^t \bar m_0$. \Comment{We normalize $\hat m$ so that we have $\hat m \ge m$.}
\Else \State $\hat m = \hat m/2$.
\EndIf
\State \Return $\hat m$.
\end{algorithmic}
\end{algorithm}

\medskip 

\noindent \textbf{Overview of Algorithm}\ \textsc{Refine-Estimate}. Suppose we are given an initial estimate $\bar m$ satisfying $m \le \bar m \le (1+\alpha) m$ for some unknown approximation factor $\alpha$ satisfying $\epsilon \le \alpha \le \binom{n}{2}$. We set the threshold value for recovering a vertex at a level $j$ as $\frac{\bar m}{\gamma^{\boundary{j}}} \cdot \frac{c_2\epsilon^2}{\log n} $ where $c_2$ is a constant. So, when a vertex $v$ with degree estimate $\hat d_j(v)$ (obtained from Algorithm~\ref{alg:estimator}) satisfies $\hat d_j(v) \ge \frac{\bar m}{\gamma^{\boundary{j}}} \cdot \frac{c_2\epsilon^2}{\log n}$, we set the level of recovery $\hat \ell(v) = j$ and recovered flag $r(v)= 1$. From construction, we can observe that once a vertex is recovered at a particular level it is not available to be recovered at higher level later. Our estimator is the summation of terms $\gamma^{\boundary{\hat \ell(v)}} \cdot \hat d(v)$ for every $v$ satisfying $r(v) = 1$. We normalize $\hat m$ to ensure that the final estimate returned satisfies $m \le \hat m$ (see the full version for additional details). 

Using Bernstein's inequality, we argue that in iteration $t$, we can improve the approximation factor of the previous estimate $\bar m_{t-1}$ by a multiplicative factor of $\epsilon$ in the new estimate $\bar m_t$. After 
 $T = O(\log_{1/\epsilon} \log  n)$ iterations, the edge estimate will be a $(1\pm \epsilon)$-relative error approximation satisfying:
 
\begin{theorem}\label{thm:edge_estimate}
Given a graph $G$ with $n$ nodes and $m$ edges, there is an algorithm that makes $O(\epsilon^{-5}\log^{5} n\log^6(\log n))$ non-adaptive BIS queries to $G$ and returns an estimate $\hat m$ satisfying:  $m(1-\epsilon) \le \hat m \le m (1+ \epsilon), \textup{with probability at least }3/5$.
\end{theorem}

\subsubsection{Proof of Theorem~\ref{thm:edge_estimate}}\label{subsubsec:edge_estimation}

First, we show that our degree estimates are calculated accurately at every level with constant probability of success.
\begin{claim}\label{cl:degreelevels_low}
With probability $3/4$, for all levels $j \in \{1, 2, \cdots, L\}$, we have:
\[ d(S_j) \le \frac{8m \cdot L}{\gamma^{\boundary{j}}}.\]
\end{claim}
\begin{proof}

As every vertex is included in $S_j$ with probability $1/\gamma^{\boundary{j}}$, we get: 
\[\E[d(S_j)] = \frac{\sum_{v \in V} d(v)}{\gamma^{\boundary{j}}} = \frac{2m}{\gamma^{\boundary{j}}}\]
Therefore, by Markov's Inequality, $\Pr[d(S_j) \ge {8m \cdot L}/{ \gamma^{\boundary{j}}} ] \le {1}/{(4L)}$. Taking a union bound over all the levels, with probability at least 3/4, $$d(S_j) \le {8m \cdot L}/\gamma^{\boundary{j}} \le {8m \cdot \log n}/\gamma^{\boundary{j}} \textup{for every level }j \in [L].$$
\end{proof}

\noindent Combining Claim~\ref{cl:degreelevels_low} and Lemma~\ref{lem:approxDeg}, for sufficiently large $n$, we have:
\begin{corollary}\label{cor:degree_estimates}
The degree estimates returned by Algorithm~\ref{alg:count} for each sampling level $j \in [L]$, satisfy the following with probability at least $0.70$:
$$(1-\epsilon) d(v) \le \hat d_j(v) \le d(v) + \frac{c_1 \epsilon^3 m}{\gamma^{\boundary{j}} \log n}  \ \quad \forall v \in S_j.$$
\end{corollary}
\begin{proof}
From Lemma~\ref{lem:approxDeg}, we have that, for every $j \in [L]$, with probability at least $1-O(1/\log n)$, the approximate degree estimates satisfy: 
$$(1-\epsilon) d(v) \le \hat d_j(v) \le d(v) + \frac{\epsilon^3 }{\log^2 n} \cdot d(S_j)  \ \quad \forall v \in S_j.$$

From Claim~\ref{cl:degreelevels_low}, we know that $d(S_j) \le 8m \cdot \log n/\gamma^{\boundary{j}}$, for every level $j \in [L]$, with probability at least $3/4$. Combining both of them, we have the claim about the approximate degree estimates.

Using union bound, we have that the total failure probability is at most $1/4 + O(L \cdot 1/\log n) \le 0.30$, as $L = O(\log n)$. Hence, the corollary.
\end{proof}

For each vertex $v \in S$ for some subset $S \subseteq V$, we associate a level $\ell(v)$ such that the \emph{actual} degree of $v$ is a large fraction of the total degree of $S$, i.e., $\ell(v)$ is the minimum $j \in \{0, 1, \cdots, L\}$ satisfying $d(v) \ge \frac{\bar m}{\gamma^{\boundary{j}}} \cdot \frac{c_2 \epsilon^2}{\log n}$. The value $\frac{\bar m}{\gamma^{\boundary{j}}} \cdot \frac{c_2 \epsilon^2}{\log n}$ is called threshold for level $\ell(v)$, and it depends on the estimate $\bar m$ for the number of edges $m$. 
\begin{definition}[Actual Level]\label{def:actual_level}
For every vertex $v \in V$, we associate a level $\ell(v)$ defined as $$\ell(v) = \arg \min_{j \in \{0, 1, \cdots, L\}} d(v) \ge \frac{\bar m}{\gamma^{\boundary{j}}} \cdot \frac{c_2 \epsilon^2}{\log n}.$$
\end{definition}

\noindent The vertices that lie close to the threshold of a level and within a $\gamma$-multiplicative factor are called the boundary vertices and are defined as below:

\begin{definition}[Boundary vertices]
The vertices closer to the \emph{boundary} are denoted by the set:
\begin{align*}
    \Vboundary &= \{ v \mid d(v) \in \left[\frac{\bar m \cdot c_2 \epsilon^2}{\gamma^{\boundary{\ell(v)}}\log n} , \frac{\bar m \cdot c_2 \epsilon^2}{\gamma^{\boundary{\ell(v)}-1}\log n} \right) \textup{ or }
    d(v) \in \left( \frac{\bar m \cdot c_2 \epsilon^2}{\gamma^{\boundary{\ell(v)-1}+1}\log n}, \frac{\bar m \cdot c_2 \epsilon^2}{\gamma^{\boundary{\ell(v)-1}}\log n}\right)  \}.
\end{align*}
\end{definition}

\begin{claim}\label{cl:boundary}
For any vertex $v \in V$, with probability $1-\epsilon$, there is some $j \in \{0, 1, \cdots, L\}$ such that:
\[   \frac{\bar m}{\gamma^{\mu(j)-1}} \cdot \frac{c_2 \epsilon^2}{\log n} \le d(v) \le \frac{\bar m}{\gamma^{\mu(j-1)+1}} \cdot \frac{c_2 \epsilon^2}{\log n}.\] 
In other words, $\Pr[v \in \Vboundary] \le \epsilon$.
\end{claim}
\begin{proof}
For notational convenience, let $\sigma = \frac{\bar m \cdot c_2 \epsilon^2}{\log n}$.
Note that for any $v \in V$ there is some $j$ such that $ \frac{\sigma}{\gamma^{\boundary{j}}} \le d(v) < \frac{\sigma}{\gamma^{\boundary{j-1}}}$. We claim that every such vertex will not lie close to the edges of the interval $\left [\frac{\sigma}{\gamma^{\boundary{j}}}, \frac{\sigma}{\gamma^{\boundary{j-1}}}\right)$, i.e., $d(v) \not \in  \left [\frac{\sigma}{\gamma^{\mu(j)}}, \frac{\sigma}{\gamma^{{\mu(j)}-1}}\right)$ and $d(v) \not\in  \left [\frac{\sigma}{\gamma^{\mu(j-1)+1}}, \frac{\sigma}{\gamma^{\mu(j-1)}}\right )$. We will show that both events occur with probability at most $1/B$, giving the claim via a union bound.

For any $v$, there is a unique $i$ such that $d(v) \in \left [\frac{\sigma}{\gamma^i}, \frac{\sigma}{\gamma^{i-1}}\right)$. Thus, the claim only fails to hold if $i = \mu(j)$ for some $j$ or $i = \mu(j-1)+1$ for some $j$. For the first case, when $i = \mu(j) = j \cdot B-s$ for some $j$ is satisfied only if $s = j\cdot B-i$, which occurs with probability $1/B$ since $s$ is selected uniformly at random from $\{0,1,\ldots,B-1\}$. Similarly, $i = \mu(j-1)+1 = (j-1) \cdot B + 1 - s$ only if $s = (j-1)\cdot B + 1 - i$, which again occurs with probability $1/B$. Using union bound, we have:
\[ \Pr\left[d(v) \in \left [\frac{\sigma}{\gamma^{\mu(j)}}, \frac{\sigma}{\gamma^{{\mu(j)}-1}}\right) \textup{ or }d(v) \in  \left [\frac{\sigma}{\gamma^{\mu(j-1)+1}}, \frac{\sigma}{\gamma^{\mu(j-1)}}\right )\right] \leq \frac{2}{B} = \epsilon.\]

Hence, the claim.
\end{proof}

\begin{definition}[Recovered Level]
A vertex $v$ is recovered at level $\hat \ell(v)$ iff \[ \hat \ell(v) = \arg \min_{j \in [L]} \hat d_j(v) \ge  \frac{\bar m}{\gamma^{\boundary{j}}} \cdot \frac{c_2 \epsilon^2}{\log n}. \]
\end{definition}

\medskip
\noindent We associate the following sets with the set of recovered vertices:
\[ \mathcal R = \{ v \in V \mid r(v) = 1\}, \Rbad = \{ v \in \mathcal R \mid \hat \ell(v) \neq \ell(v) \}, \textup{ and }\Rboundary = \RR \cap \Vboundary.\]

Here, $\Rbad$ represents set of recovered vertices $v$ at a level $\hat \ell(v)$ different from $\ell(v)$. Recall that $\ell(v)$ represents the level at which the vertex $v$ is recovered if we knew the degree $d(v)$ exactly.

Using the next claim, we argue that if $v$ is included in the set of sampled vertices at level $\ell (v)$, i.e., $v \in S_{\ell(v)}$, then, it will be recovered at that level, provided degree estimates satisfy Corollary~\ref{cor:degree_estimates}.

\begin{claim}\label{cl:recovery}
Suppose $v \in S_{\ell(v)}$ and $v \not\in \Vboundary$ satisfying:
\[ (1-\epsilon) d(v) \le \hat d_{\hat \ell(v)}(v) \le d(v) +  \frac{c_1\epsilon^3 \cdot m}{\log n \cdot \gamma^{\boundary{\hat \ell(v)}}}, \text{then, we have } \hat \ell(v) = \ell(v).\]
\end{claim}
\begin{proof}
As $v \not\in \Vboundary$, and $v \in S_{\ell(v)}$, from the definition of boundary vertices, we have:
\begin{align*}
   \frac{\bar m}{\gamma^{\boundary{\ell(v)}-1}} \cdot \frac{c_2 \epsilon^2}{\log n} \le d(v) \le \frac{\bar m}{\gamma^{\boundary{\ell(v)-1}+1}} \cdot \frac{c_2 \epsilon^2}{\log n}.
\end{align*}

This implies:
\begin{align*}
\hat d_{\ell(v)}(v) \ge (1-\epsilon)d(v) &\ge (1-\epsilon) \cdot \frac{\bar m}{\gamma^{\boundary{\ell(v)}-1}} \cdot \frac{c_2 \epsilon^2}{\log n}  \\
 &= (1-\epsilon) \gamma \cdot\frac{\bar m}{\gamma^{\boundary{\ell(v)}}} \cdot \frac{c_2\epsilon^2}{\log n}\\
&=  \frac{\bar m}{\gamma^{\boundary{\ell(v)}}} \cdot \frac{c_2\epsilon^2}{\log n},
\end{align*}

as $\gamma = {1}/{(1-\epsilon)}$, and so in Algorithm~\rEstimate (Alg.~\ref{alg:refineEst}), $v$ will be recovered, and $\hat \ell(v) = \ell(v)$ as long as it hasn't been recovered at a prior level.

At any prior level $j \leq \ell(v)-1$, from Lemma~\ref{lem:approxDeg}, we have:
\begin{align*}
\hat d_j(v) &\le d(v) +  \frac{c_1\epsilon^3 \cdot m}{\log n \cdot \gamma^{\boundary{j}}}\\
&\le   \frac{\bar m}{\gamma^{\boundary{\ell(v)-1}+1}} \cdot \frac{c_2 \epsilon^2}{\log n} +  \frac{c_1\epsilon^3 \cdot m}{\log n \cdot \gamma^{\boundary{j}}}\\
&\le \frac{\bar m}{\gamma^{\boundary{j}}} \cdot \frac{c_2\epsilon^2}{\log n} \cdot \left( \frac{1}{\gamma \cdot \gamma^{(\ell(v)-1-j) \cdot B}} + \frac{ c_1 \epsilon}{c_2}\right)\\
&\le \frac{\bar m}{\gamma^{\boundary{j}}} \cdot \frac{c_2\epsilon^2}{\log n} \left( \frac{1}{\gamma} +  \frac{c_1\epsilon}{c_2} \right)\\
&= \frac{\bar m}{\gamma^{\boundary{j}}} \cdot \frac{c_2\epsilon^2}{\log n} \left( 1-\epsilon +  \frac{c_1\epsilon}{c_2} \right)\\
&< \frac{\bar m}{\gamma^{\boundary{j}}} \cdot \frac{c_2\epsilon^2}{\log n},
\end{align*}
as long as we set $c_1 < c_2$. Thus, $v$ will be rejected at any level $j < \ell(v)$. 
\end{proof}

The following corollary is immediate from the previous Claim~\ref{cl:recovery}, as every vertex that is not at the boundary is recovered at the actual level.
 \begin{corollary}
 If all the degree estimates of sampled vertices at every level are {good} approximations, i.e., satisfy the Corollary~\ref{cor:degree_estimates}, then, $\Rbad \subseteq \Rboundary \subseteq \Vboundary$.
 \end{corollary}
 
For a vertex $v$ that lies in the boundary, i.e., $v \in \Vboundary$, it is possible that $v$ is recovered at a level far away from $\ell(v)$. Using the next claim, we argue that it will be recovered in the adjacent levels if it has not been recovered at $\ell(v)$.

\begin{claim}\label{cl:bad_boundary_vertices}
Suppose $v \in S_{\ell(v)+1}$ and $v \in \Vboundary$ satisfying: $$ (1-\epsilon) \le \hat d_{\hat \ell(v)}(v) \le d(v) + \frac{c_1\epsilon^3 \cdot m}{\log n \cdot \gamma^{\boundary{\hat \ell(v)}}}, \text{ then, we have }\hat \ell(v) \in \{ \ell(v) + 1, \ell(v), \ell(v) -1 \}.$$
\end{claim}
\begin{proof}
For notational convenience, let $\sigma = \frac{\bar m \cdot c_2 \epsilon^2}{\log n}$. As $v \in S_{\ell(v)+1}$, we have $v \in S_{\ell(v)-1}$ and $v \in S_{\ell(v)}$ from construction.

First, we observe that $\hat \ell(v) \le \ell(v) + 1$, because,  \[\hat d_{\hat \ell(v)}(v) \ge (1-\epsilon) d(v) \ge  (1-\epsilon)\frac{\sigma}{\gamma^{\boundary{\ell(v)}}} \ge (1-\epsilon) \gamma^B \frac{\sigma}{\gamma^{\boundary{\ell(v)+1}}} \ge \frac{\sigma}{\gamma^{\boundary{\ell(v)+1}}}.\]

The last inequality follows from $(1-\epsilon) \gamma^B \ge (1-\epsilon)(1+\epsilon)^B \ge 3-3\epsilon \ge 1$, when $\epsilon \le \frac{2}{3}$. 

As $v \in \Vboundary$, we have the following cases:
\begin{itemize}
\item[(a)] $d(v) \in \left( \frac{\sigma}{\gamma^{\mu(\ell(v)-1)+1}}, \frac{\sigma}{\gamma^{\mu(\ell(v)-1)}}\right)$. We use proof by contradiction. Suppose $\hat \ell(v) \le \ell(v)-2$. Then: 
\begin{align*}
\hat d_{\hat \ell(v)}(v) &\le d(v) +  \frac{c_1\epsilon^3 \cdot m}{\log n \cdot \gamma^{\boundary{\hat \ell(v)}}}\\
&<   \frac{\bar m}{\gamma^{\boundary{\ell(v)-1}}} \cdot \frac{c_2 \epsilon^2}{\log n} +  \frac{c_1\epsilon^3 \cdot m}{\log n \cdot \gamma^{\boundary{\hat \ell(v)}}}\\
&= \frac{\bar m}{\gamma^{\boundary{\hat \ell(v)}}} \cdot \frac{c_2\epsilon^2}{\log n} \cdot \left( \frac{1}{\gamma^{(\ell(v)-1-\hat \ell(v)) \cdot B}} + \frac{ c_1 \epsilon}{c_2}\right)\\
&\le \frac{\bar m}{\gamma^{\boundary{\hat \ell(v)}}} \cdot \frac{c_2\epsilon^2}{\log n} \left( \frac{1}{\gamma^B} +  \frac{c_1\epsilon}{c_2} \right)\\
&= \frac{\bar m}{\gamma^{\boundary{\hat \ell(v)}}} \cdot \frac{c_2\epsilon^2}{\log n} \left( (1-\epsilon)^B +  \frac{c_1\epsilon}{c_2} \right)\\
&\le \frac{\bar m}{\gamma^{\boundary{\hat \ell(v)}}} \cdot \frac{c_2\epsilon^2}{\log n} \left( e^{-2} +  \frac{c_1\epsilon}{c_2} \right) \le \frac{\bar m}{\gamma^{\boundary{\hat \ell(v)}}} \cdot \frac{c_2\epsilon^2}{\log n}.
\end{align*}
The last inequality follows because $c_1 \le c_2(1-e^{-2})$. Therefore, $\hat \ell(v) > \ell(v) - 2$.
\item[(b)]  $d(v) \in \left[ \frac{\sigma}{\gamma^{\mu(\ell(v))}} , \frac{\sigma}{\gamma^{\mu(\ell(v))-1}} \right).$ Using a similar argument, we obtain that $\hat \ell(v) \ge \ell(v)$. 

For the sake of contradiction, let $\hat \ell(v) \le \ell(v)-1$. Then:
\begin{align*}
\hat d_{\hat \ell(v)}(v) &\le d(v) +  \frac{c_1\epsilon^3 \cdot m}{\log n \cdot \gamma^{\boundary{\hat \ell(v)}}}\\
&<   \frac{\bar m}{\gamma^{\boundary{\ell(v)}-1}} \cdot \frac{c_2 \epsilon^2}{\log n} +  \frac{c_1\epsilon^3 \cdot m}{\log n \cdot \gamma^{\boundary{\hat \ell(v)}}}\\
&= \frac{\bar m}{\gamma^{\boundary{\hat \ell(v)}}} \cdot \frac{c_2\epsilon^2}{\log n} \cdot \left( \frac{\gamma}{\gamma^{(\ell(v)-\hat \ell(v)) \cdot B}} + \frac{ c_1 \epsilon}{c_2}\right)\\
&\le \frac{\bar m}{\gamma^{\boundary{\hat \ell(v)}}} \cdot \frac{c_2\epsilon^2}{\log n} \left( \frac{2}{\gamma^B} +  \frac{c_1\epsilon}{c_2} \right)\\
&= \frac{\bar m}{\gamma^{\boundary{\hat \ell(v)}}} \cdot \frac{c_2\epsilon^2}{\log n} \left( 2(1-\epsilon)^{B} +  \frac{c_1\epsilon}{c_2} \right)\\
&\le \frac{\bar m}{\gamma^{\boundary{\hat \ell(v)}}} \cdot \frac{c_2\epsilon^2}{\log n} \left( 2e^{-2} +  \frac{c_1\epsilon}{c_2} \right) \le \frac{\bar m}{\gamma^{\boundary{\hat \ell(v)}}} \cdot \frac{c_2\epsilon^2}{\log n}.
\end{align*}
The last inequality follows because $c_1 \le c_2(1-2e^{-2})$. Therefore, $\hat \ell(v) > \ell(v) - 1$.
\end{itemize}

Therefore, we have $\ell(v) - 2 < \hat \ell(v) \le \ell(v) + 1$.
\end{proof}

\begin{definition}[Random variables]
 Let $\hat X(v)$ be the random variable with $\hat X(v) = \gamma^{\boundary{\hat \ell(v)}} \cdot \hat d(v)$ if $v$ is recovered at level $\hat \ell(v)$ and $\hat X(v) = 0$ otherwise. We define $X(v)$ similarly, assuming we run Algorithm~\ref{alg:refineEst} with exact degrees, i.e., $X(v) = \gamma^{\boundary{\ell(v)}} \cdot d(v)$ if $v$ is recovered at its actual level $\ell(v)$ and $X(v) = 0$ otherwise.

\end{definition}

\medskip
In the analysis that follows next, we will argue that for most of the vertices in $\mathcal{R}$, except for those in $\Rboundary$, the $\hat X(v)$ and $X(v)$ are close to each other, i.e., $1\pm \epsilon$ approximations of each other. Separately, we show that the contribution of $\sum_{v \in \Rboundary} \hat X(v)$ is small. As $X(v)$ values do not contain any degree approximations, they are easier to handle and we will show concentration for $\sum_{v \in \RR} X(v)$. As a result, the concentration will also hold for the actual estimate $\sum_{v \in \RR} \hat X(v)$.

\medskip
Throughout the remaining section, unless explicitly stated otherwise, we will assume that the degrees are \emph{good} approximations. Formally stated, we define $\mathcal E_0$ as the event indicating all the degree estimates at every sampling level satisfy Corollary~\ref{cor:degree_estimates}. Note that $\Pr[\mathcal E_0] \ge 0.70$ (from Corollary~\ref{cor:degree_estimates}).

\begin{claim}\label{cl:badv}
\text{With probability at least }$1-1/11\log \log n$, we have: \[\sum_{v \in \Rboundary} \hat X(v) \le 572 \epsilon m \log \log n \quad \text{ and }\sum_{v \in \Rboundary} X(v) \le 44\epsilon m \log \log n. \]
\end{claim}
\begin{proof}

For the sake of brevity, we omit that all the expected values include conditioning on the event $\mathcal E_0$.

Consider a vertex $v \in \Rboundary$. We have that the probability a vertex $v \in \Rboundary$ is recovered at a level $\hat \ell(v)$ satisfying $\hat \ell(v) \in \{ \ell(v)+1, \ell(v), \ell(v)-1\}$ (From Claim~\ref{cl:bad_boundary_vertices}):
\begin{align*}
    \Pr[v \in \Rboundary ] &= \Pr[v \in \Vboundary] \cdot \Pr[v \in \RR \mid v \in \Vboundary]\\
    &=  \sum_{\hat \ell(v) \in \{ \ell(v)+1, \ell(v), \ell(v)-1\}} \Pr[v \in \Vboundary]\cdot \Pr[v \in S_{\hat \ell(v)}].
\end{align*}

Therefore, we have:
\begin{align*}
    \E[\hat X(v)] &= \sum_{\hat \ell(v) \in \{ \ell(v)+1, \ell(v), \ell(v)-1\}}  \gamma^{\boundary{\hat \ell(v)}} \hat d_{\hat \ell(v)} \Pr[v \in \Vboundary]\cdot \Pr[v \in S_{\hat \ell(v)}]\\
    &\le \epsilon \cdot \sum_{\hat \ell(v) \in \{ \ell(v)+1, \ell(v), \ell(v)-1\}}   \hat d_{\hat \ell(v)}\\
     &\le 3\epsilon \cdot \left( d(v) + \frac{c_1 \epsilon^3 m}{\gamma^{\boundary{\ell(v)-1}} \log n}\right)\quad \mbox{(using Corollary~\ref{cor:degree_estimates})} \\
    &= 3\epsilon \cdot \left(d(v) + \left(\frac{1}{1-\epsilon} \right)^B \cdot \frac{c_1 \epsilon^3 m}{\gamma^{\boundary{\ell(v)}} \log n}\right)\\
    &\le 3\epsilon \cdot \left(d(v) + 5\epsilon d(v) \right), \textup{ because } d(v) \ge \frac{c_2 \epsilon^2 \bar m}{\gamma^{\boundary{\ell(v)}} \log n } \ge \frac{c_1 \epsilon^2 m}{\gamma^{\boundary{\ell(v)}} \log n }\\
    &\le 18 \epsilon \cdot d(v)\\
    \Rightarrow \E\left[\sum_{v \in \Rboundary} \hat X(v) \right] &\le 18 \epsilon \cdot \sum_{v \in V} d(v) \le 36\epsilon m.
\end{align*}

Using Markov's inequality, we have $\sum_{v \in \Rbad} \hat X_v \le 572\epsilon m \log n$, with probability $\ge 1-1/22\log \log n$.

Similarly, we can bound the sum:
\begin{align*}
    \E\left[\sum_{v \in \Rboundary} X(v)\right] \le \sum_{v \in V} \epsilon \cdot d(v) \le 2\epsilon m. 
\end{align*}

Using Markov's inequality, we have, with probability $1-1/22 \log \log n$, $\sum_{v \in \Rbad} X(v) \le 44\epsilon m \log \log n$. Taking a union bound for both the events, gives us the claim.
\end{proof}

\begin{claim}\label{clm:exp}
$\E[X(v)] = d(v) \ \forall v \in V$. Also, $\E[\sum_{v \in V} X(v)] = 2m$.
\end{claim}
\begin{proof}
By Claim \ref{cl:recovery}, $X(v)$ is nonzero  which requires that $v \in S_{\ell(v)}$. As $v$ is included in $S_{\ell(v)}$ with probability $1/\gamma^{\boundary{\ell(v)}}$. Therefore, we have:
\begin{align*}
    \E[X(v)] =  \gamma^{\boundary{\ell(v)}} \cdot d(v) \cdot 1/\gamma^{\boundary{\ell(v)}} &= d(v)\\
    \E \left[\sum_{v \in V} X(v)\right] = \sum_{j \in [L]}\sum_{v \in V \cap S_{j}} \E\left[X(v) \right] &= \sum_{v \in V} d(v) = 2m.
\end{align*}
\end{proof}

\begin{claim}\label{cl:goodv}
 $ \left |\sum_{v\in \mathcal{R} \setminus \Rboundary} \hat X(v) - \sum_{v\in \mathcal{R} \setminus \Rboundary}  X(v) \right | \le \epsilon \cdot \sum_{v\in \mathcal{R} \setminus \Rboundary}  X(v)$.
\end{claim}
\begin{proof}
Consider a vertex $v \in \RR \setminus \Rboundary$. From Claim~\ref{cl:recovery}, we have $\hat \ell(v) = \ell(v)$ provided $v \in S_{\ell(v)}$. As we have already conditioned on the event $\mathcal E_0$, from Corollary~\ref{cor:degree_estimates}, we have:
\begin{align*}
    (1-\epsilon) d(v) &\leq \Hat d_{\ell(v)}(v) \leq d(v) + \frac{c_1 \epsilon^3 m}{\gamma^{\boundary{\ell(v)}} \log n}\\
 (1-\epsilon) \gamma^{\boundary{\ell(v)}} d(v) &\leq \gamma^{\boundary{\ell(v)}}\Hat d_{\ell(v)}(v) \leq \gamma^{\boundary{\ell(v)}}d(v) + \gamma^{\boundary{\ell(v)}} \cdot \frac{c_1 \epsilon^3 m}{\gamma^{\boundary{\ell(v)}} \log n}\\
        \Rightarrow (1-\epsilon) X(v) &\leq \Hat X(v) \leq X(v) + \frac{c_1\epsilon^3 m}{\log n}.
\end{align*}

As $\bar m \ge m$ and $c_2 > c_1$, we have: 
\[d(v) \ge \frac{\Bar{m} \cdot c_2 \epsilon^2}{\gamma^{\boundary{\ell(v)}} \log n} \ge \frac{m \cdot c_2 \epsilon^2}{\gamma^{\boundary{\ell(v)}} \log n}\]
\[\Rightarrow \frac{m \cdot c_1\epsilon^3}{\log n} \le \frac{{m} \cdot c_2 \epsilon^3}{\log n} \le \gamma^{\boundary{\ell(v)}} \cdot  \epsilon d(v) = \epsilon X(v). \]

Therefore, 
\[ (1-\epsilon) X(v) \leq \Hat X(v) \leq (1+\epsilon) X(v).\]

Thus, if $\hat X(v)$ is nonzero, $ (1-\epsilon) \hat X(v) \leq X(v) \leq (1+\epsilon) \hat X(v)$, which gives the claim after summing up over all the terms. 
\end{proof}

\begin{claim}\label{clm:bound}
The variables $X(v) \ \forall v \in V$ are independent and bounded given by $$X(v) \le \frac{2 \bar m \cdot c_2 \epsilon^2}{\log n}$$
\end{claim}
\begin{proof}
Since the vertices are included independently in the sets $S_0,\ldots,S_L$, the independence of $X(v)$ follows immediately. Additionally, since $\ell(v)$ is the smallest $j \in \{0, 1, \ldots L\}$ for which $d(v) \ge \frac{\bar m}{\gamma^{\boundary j}} \cdot \frac{c_2 \epsilon^2}{\log n}$ (See the definition of $\ell(v)$ in Claim \ref{cl:recovery}), we obtain $d(v) \le \frac{\bar m}{\gamma^{\boundary{\ell(v)-1}}} \cdot \frac{c_2 \epsilon^2}{\log n}$. Thus, if $X(v)$ is nonzero,
\begin{align*}
X(v) = \gamma^{\boundary{\ell(v)}} \cdot d(v)  \le \frac{\gamma^B \bar m \cdot c_2 \epsilon^2}{\log n} \le  \frac{1}{(1-\epsilon)^B} \cdot \frac{\bar m \cdot c_2 \epsilon^2}{\log n} \le \frac{2 \bar m \cdot c_2 \epsilon^2}{\log n}.
\end{align*}
\end{proof}

\noindent Combining Claims \ref{clm:exp} and \ref{clm:bound} we obtain:
\begin{claim}\label{clm:final}
With probability at least $1-1/n$, $|\sum_v X(v) - 2m| \le \max\{\epsilon, \epsilon \cdot \alpha \} \cdot 2m$.
\end{claim}
\begin{proof}
By Claims \ref{clm:bound} and \ref{clm:exp}, we have:
\begin{align*}
    \sum_{v \in V} \E[(X(v) - d(v))^2] &= \sum_v \E[X(v)^2] - 2d(v)\E[X(v)] + d(v)^2\\
                        &\leq \sum_v \E[X(v)^2]\\
                        &\leq \frac{2 \bar m \cdot c_2 \epsilon^2}{\log n} \cdot \sum_{v \in V} \E[X(v)] = \frac{2\bar m \cdot 2m \cdot c_2 \epsilon^2}{\log n}.
\end{align*}

If $\alpha > 1$, then, $m \le \bar m \le m(1+\alpha)\le 2m \alpha$. From Bernstein's inequality (Lemma~\ref{lem:bernstein}), we have:
\begin{align*}\Pr \left [ \left |\sum_v X(v) - 2m \right | \ge  \epsilon' \alpha \cdot m \right ] &\le \exp \left (-\frac{\epsilon'^2 \alpha^2 \cdot m^2}{\frac{2 \cdot 2\bar m \cdot 2m \cdot c_2 \epsilon^2}{\log n} + \frac{2}{3} \cdot \frac{2\bar m \cdot  c_2 \epsilon^2}{\log n} \cdot {m \epsilon' \alpha}  }\right )\\
&\le \exp \left (-\frac{\epsilon'^2 \alpha^2 \cdot m^2 \log n}{{2 \cdot 4m \alpha \cdot 2m \cdot c_2 \epsilon^2} + \frac{2}{3} \cdot {4m \alpha \cdot  c_2 \epsilon^2} \cdot {m \epsilon' \alpha}  }\right)\\
&= \exp \left (-\frac{\epsilon'^2 \alpha^2 \log n}{{ 16 c_2 \epsilon^2 \alpha} + \frac{8}{3} \cdot {c_2 \epsilon^2} \cdot {\epsilon' \alpha^2}  }\right) \le \frac{1}{n}, \  \textup{because},\\
\frac{\epsilon'^2 \alpha^2}{{ 16 c_2 \epsilon^2 \alpha} + \frac{8}{3} \cdot {c_2 \epsilon^2} \cdot {\epsilon' \alpha^2}  }&\ge 1, \textup{as } 1 \ge \epsilon' \ge \epsilon \left( \frac{56c_2}{3}  \right)^{1/2}.  
\end{align*}

If $\alpha \in [\epsilon',1]$, then, $\bar m \le 2m$. From Bernstein's inequality (Lemma~\ref{lem:bernstein}), we have:
\begin{align*}
\Pr \left [ \left |\sum_v X(v) - 2m \right | \ge  \epsilon' \cdot m \right ] 
&\le \exp \left (-\frac{\epsilon'^2 \cdot m^2}{\frac{2 \cdot 2\bar m \cdot 2m \cdot c_2 \epsilon^2}{\log n} + \frac{2}{3} \cdot \frac{2\bar m \cdot  c_2 \epsilon^2}{\log n} \cdot {m \epsilon'}  }\right )\\
&\le \exp \left (-\frac{\epsilon'^2 \cdot m^2 \log n}{{2 \cdot 4m  \cdot 2m \cdot c_2 \epsilon^2} + \frac{2}{3} \cdot {4m \cdot  c_2 \epsilon^2} \cdot {m \epsilon'}  }\right )\\
&\le \exp \left (-\frac{\epsilon'^2 \log n}{{ 64 c_2 \epsilon^2 } + \frac{16}{3} \cdot {c_2 \epsilon^2} }\right) \le \frac{1}{n}, \  \textup{because},\\
\frac{\epsilon'^2}{{ 64 c_2 \epsilon^2 } + \frac{16}{3} \cdot {c_2 \epsilon^2}  }&\ge 1 \textup{, as } 1 \ge \epsilon' \ge \epsilon \left( \frac{208c_2}{3}  \right)^{1/2}.  
\end{align*}

Combining both the statements, and scaling $\epsilon$ gives us the lemma.
\end{proof}

Note that Claim \ref{clm:final} basically gives us what we want, after adjusting constants on $\alpha$ and scaling up our estimate appropriately. We formalize this using the below lemma:

\begin{lemma}\label{lem:approx_counting}
Suppose the input $\bar m$ to Algorithm \rEstimate satisfies $m \le \bar m \le m(1+\alpha)$ for some approximation factor $\epsilon \le \alpha \le \frac{ \binom{n}{2}}{m}$. Then, with probability $1-1/11\log \log n-1/n$: 
$$ 2m(1- \epsilon \log \log n - \epsilon \cdot \alpha) \le \sum_{v \in V} \hat X(v) \le 2m(1 + \epsilon \cdot \log \log n+\epsilon \cdot \alpha).$$
\end{lemma}
\begin{proof}
We note that for the vertices that are not recovered, i.e., $r(v) = 0$, we have $\hat X(v) = 0$, and therefore need to only consider vertices in $\mathcal R$. From Claim~\ref{cl:goodv}, we have: 
\[ \left |\sum_{v\in \mathcal{R} \setminus \Rboundary} \hat X(v) - \sum_{v\in \mathcal{R} \setminus  \Rboundary}  X(v) \right | \le \epsilon \cdot \sum_{v\in \mathcal{R} \setminus  \Rboundary}  X(v)\]

Combining it with
$\sum_{v \in  \Rboundary} \hat X(v) \le 572\epsilon m\log \log n$ from Claim~\ref{cl:badv}, we have:
\begin{align*}
  \sum_{v \in \mathcal R} \hat X(v) &\le (1+\epsilon) \sum_{v \in \RR \setminus \Rboundary} X(v) + \sum_{v \in  \Rboundary} \hat X(v)\\
    &\le (1+\epsilon) \sum_{v\in V}  X(v)+ 572\epsilon m \log \log n\\
     &\le 2m (1+\epsilon \cdot \alpha + \epsilon \cdot \log \log n), 
\end{align*}
where the last step follows by scaling $\epsilon$ (with a constant) appropriately. From Claim~\ref{cl:badv}, we have  $\sum_{v \in \Rboundary} X(v) \le 44 \epsilon m \log \log n$. Therefore, we get:
\[ \sum_{v \in V} X(v) \le \sum_{v \in \RR \setminus \Rboundary} X(v) + 44\epsilon m \log \log n. \]

Similarly, we have:
\begin{align*}
     \sum_{v \in \mathcal{R}} \hat X(v)&\ge (1-\epsilon) \sum_{v\in \mathcal{R} \setminus  \Rboundary}  X(v) + \sum_{v \in  \Rboundary} \hat X(v) \\
    &\ge (1-\epsilon) \sum_{v\in \mathcal{R} \setminus \Rboundary}  X(v)\\
    &\ge (1-\epsilon) \cdot \sum_{v\in V} X(v) - (1-\epsilon)44\epsilon m \log \log n\\
    &\ge 2m(1- \epsilon \log n - \epsilon \cdot \alpha) \ge 2m(1-  \alpha \cdot \epsilon \log \log n).
\end{align*}

The last step follows by scaling $\epsilon$ appropriately. Using union bound, we get the final probability claim.  Hence, the lemma.
\end{proof}

In Algorithm~\rEstimate, we start with an estimate $\bar m_0$ for the number of edges, due to Beame et al.\cite{beame2020edge}. The estimate for the number of edges $\bar m_0$ satisfies $m \le \bar m_0 \le m(1+\alpha_0)$, for some $\alpha_0$. As $m \le \bar m_0 \le O(m \log^2 n)$, the approximation factor $\alpha_0$ could be as large as $O(\log^2 n)$. In Lemma~\ref{lem:approx_counting}, we showed that we can improve our estimator by $\epsilon \cdot (\alpha + \log \log n)$ multiplicative factor.  We call this multiplicative improvement as \emph{refinement}. In the next theorem, we argue that our Algorithm~\ref{alg:estimator} which performs repeated refinements results in a $(1\pm \epsilon)$-approximation.

\begin{theorem}[Theorem~\ref{thm:edge_estimate} restated]\label{thm:edge_estimate_app}
Given a graph $G$ with $n$ nodes and $m$ edges, there is an algorithm that makes $O(\epsilon^{-5}\log^{5} n\log^6(\log n))$ non-adaptive BIS queries to $G$ and returns an estimate $\hat m$ satisfying:  $m(1-\epsilon) \le \hat m \le m (1+ \epsilon), \textup{with probability at least }3/5$.
\end{theorem}
\begin{proof}
    We denote $\epsilon' = O(\epsilon \log \log n)$. As $m \le \bar m_0 \le O(m \log^2 n)$, the approximation factor $\alpha_0$ could be as large as $O(\log^2 n)$. From Lemma~\ref{lem:approx_counting}, we have that each refinement improves the approximation factor to $m (1+\epsilon \cdot \alpha + \epsilon \cdot \log n ) \le m(1+\epsilon'\alpha_0) \le m + \epsilon' m \alpha_0$. Therefore, after $O(\log_{{1}/{\epsilon'}} \log n)$ \emph{refinements}, we expect the upper bound in the approximation to reduce from $\alpha_0 \cdot m$ to $O(\epsilon' \cdot m)$. However, each \emph{refinement} worsens the lower bound from $m$ to $m(1-\epsilon' \cdot \alpha)$. In order to maintain the invariant that the input to Algorithm~\ref{alg:refineEst} always satisfies $\bar m_1 = \hat m \ge m$, we normalize it by adding $\epsilon' m_0$.  This implies that $\hat m \le m + 2 \epsilon' m \alpha_0$, and the new approximation factor $\alpha_1 \le 2\epsilon' \alpha_0$. Continuing this, after $T-1 = 3\log_{1/\epsilon'} \log n$ refinements, we will have $\hat m_T \le m + (2\epsilon')^T m \alpha_0$. By scaling $\epsilon = O(\epsilon/\log \log n)$, we have, $\epsilon' \le 1/c$ for some integer constant $c > 2$ and $\hat m_T \le m + \epsilon \cdot m$. So, the final estimate $\hat m$ returned satisfies:
\[ m(1-\epsilon) \le \hat m \le m (1+ \epsilon)\]

For each level of sampling, we use Algorithm~\ref{alg:count} to return degree estimates which requires $O(\epsilon^{-5}\log^4 n\log(\log n))$ BIS queries and $O(\epsilon^{-5}\log^5 n\log^2(\log n))$ in total including all the $L$ levels (without scaling of $\epsilon$). As we have scaled by setting $\epsilon = O(\epsilon/\log \log n)$, the total number of BIS queries is $O(\epsilon^{-5}\log^5 n\log^6(\log n))$.

Recall that we have conditioned on the event $\mathcal E_0$ that our degree estimates are accurate. Using union bound, the total probability of failure across $O(\log \log n)$ refinements is:
\[\Pr[\neg \mathcal E_0] + (1/10\log \log n + 1/n) \cdot \log \log n \le 0.30 + 1/11 + \log \log n/n \le 2/5 \textup{ for sufficiently large }n.\]

Hence, the theorem.
\end{proof}

\section{Non-adaptive algorithms for uniform sampling}\label{sec:sampler}
In this section, we describe algorithms for sampling a \textit{near}-uniform edge in the graph. In Section~\ref{subsec:nbr_sampling}, we discuss connections between OR queries (Definition~\ref{def:or_query}) and BIS queries, and outline an algorithm that takes as input two disjoint subsets $L, R$ and returns a uniform vertex in $|\nbr{L} \cap R|$. Next, in Section~\ref{subsec:degree_sampling}, we use this algorithm to return uniform neighbors for every vertex in a given subset sampled from $V$. Finally, in Section~\ref{subsec:sampler}, we combine these neighbors obtained for each vertex, and return a \textit{near}-uniform sample among the edges of the graph.

\subsection{Identifying uniform neighbor of a subset of vertices}\label{subsec:nbr_sampling}

We will describe connections between OR queries (Definition~\ref{def:or_query}) and BIS queries that give us algorithms for sampling uniformly an edge from the neighborhood of a subset of vertices $L$ in another disjoint subset $R$. This is similar to Section~\ref{sec:counting}, where we discussed algorithms for counting the number of edges in the neighborhood of a set $L$ in another disjoint subset $R$.

In~\cite{assadi2021}, the authors discuss various algorithms for the well-studied single element recovery problem using OR-queries. In the single element recovery problem, we are given a boolean vector, and we want to identify a non-zero index (also called element) of the vector. 

\begin{definition}[Single Element Recovery~\cite{assadi2021}]
Given a boolean vector $x \in \{ 0, 1 \}^N$, return a non-zero element from the support of $x$, denoted by $\support{x}$.
\end{definition}

\begin{lemma}[Lemma 4.3 from~\cite{assadi2021} restated]\label{lem:ser}
Suppose $x \in \{ 0, 1 \}^N$ is a boolean vector. There is a non-adaptive randomized algorithm that recovers a uniform element $j \in \support{x}$ with probability $1-\delta$ and uses $O(\log^2 N \log(1/\delta))$ OR queries.
\end{lemma}

\medskip
\noindent \textbf{Simulating a BIS query using OR query}. We start with an  observation that any BIS query can be simulated using a single OR query. An OR query (Definition~\ref{def:or_query}) takes as input a subset $S \subseteq V \times V$ of pairs of vertices and returns if an edge of the graph is present among the subset. Therefore, a BIS query $\bis{L}{R}$ is equivalent to an OR query of the subset $S = \{ (u, v) \mid u \in L, v \in R\}$. 

\medskip
Now, we will show a connection in the other direction for the problem of single element recovery, i.e., we show that OR queries used for finding single element recovery can be simulated using appropriate BIS queries.

Suppose we are given two disjoint subsets $L, R \subseteq V$, and we want to output a neighbor of $L$ in $R$, i.e., a vertex in the set $\nbr{L} \cap R$. Intuitively, this is equivalent to finding a non-zero element (corresponds to a neighbor vertex) in the vector defined over the subset $R$ for a given subset $L$. 

\medskip
\noindent \textbf{Simulating an OR query using BIS query for finding a neighbor}. Let $x_R \in \{0, 1\}^R$ denote a vector such that $i$th element of the vector corresponds to the the $i$the vertex $v_i$ in $R$ (according to some fixed ordering of vertices). If $v_i \in \nbr{L} \cap R$ for some $i \in \{1, 2, \cdots, |R|\}$,  then, $x_R[i] = 1$, otherwise it is $0$. Let $\mathcal Q$ denotes the set of OR queries used to recover a uniform element from $x_R$ using the algorithm from Lemma~\ref{lem:ser}. Each of the OR queries $q \in Q$ is defined over a subset $R_q \subseteq \{ v_1, v_2, \cdots v_{|R|} \}$ and can be replaced with a corresponding BIS query $\bis{L}{R_q}$ with the same output. Therefore, we can restate Lemma~\ref{lem:ser} in terms of BIS queries as follows:

\begin{lemma}\label{lem:nbr_bis}
Suppose $L, R \subseteq V$ are disjoint subsets. There is a non-adaptive randomized algorithm that recovers a uniform neighbor $u \in \nbr{L} \cap R$ with probability $1-\delta$ and uses $O(\log^2 |R| \log(1/\delta))$ BIS queries.
\end{lemma}

\subsection{Identifying uniform neighbour for each vertex}\label{subsec:degree_sampling}

In this section, given a subset $S \subseteq V$, we describe an algorithm that returns a uniform neighbor for every vertex in $S$ based on the ideas from Section~\ref{subsec:nbr_sampling}.

\medskip
\noindent \textbf{Overview of Algorithm~\uneighbor}. Our algorithm extends \approxdegree by also returning a uniform neighbor for every vertex in $S$, along with the degree estimates. Consider a vertex $v$ contained in the partition $S^{ta}$. Along with the estimating the neighborhood size of the partition containing a vertex $v$, we also return a uniform neighbor from the set $\nbr{S^{ta}} \cap V \setminus S^{ta}$ using Lemma~\ref{lem:nbr_bis} from Section~\ref{subsec:nbr_sampling}. So, we will have a set of $T = O(\log n)$ neighbors for every vertex. By selecting the neighbor corresponding to the partition $S^{\tmin(v)a}$ containing $v$, we ensure that the neighbor is a uniform neighbor of $v$ with high success probability. Here, the partition $S^{\tmin(v)a}$ where $\tmin(v) \in [T]$, corresponds to the random partition with the minimum neighborhood estimate size and is used for degree estimate of $v$, i.e., $\hat d(v)$.

\begin{algorithm}[!ht]
\caption{\uneighbor: Uniform neighbor for each vertex in a given subset $S$}
\label{alg:count_sampling}
\begin{algorithmic}[1]
\Statex \textbf{Input:} Subset $S \subseteq V$.
\Statex \textbf{Output:} Degree estimates and a uniform neighbour for each vertex $v \in S$.
\State Initialize $\Hat{d}(v) \leftarrow n $ for every $v \in S$.
\For{$t$ in $\{1, 2, \ldots, T = O(\log n)\}$}
\State Consider a random partitioning of $S$ into $S^{t1}, S^{t2}, \ldots S^{t\lambda}$ where \highlight{$\lambda = O(\epsilon^{-4} \log^2 n)$}.
\For{every partition $S^{ta}$ where $a \in [\lambda]$}
\State \highlight{Let $\zz{ta}$ is sample returned using Lemma~\ref{lem:nbr_bis} where $L= S^{ta}, R= V \setminus S^{ta}$ and $\delta = O(\epsilon/\log n^4)$.}
\State $\nbrest{}(S^{ta}) \leftarrow \nbralg(S^{ta}, V \setminus S^{ta})$. 
\For{$v \in S^{ta}$}
\If{$\Hat{d}(v) > \nbrest{}(S^{ta})$}
\State $\Hat{d}(v) \leftarrow \nbrest{}(S^{ta})$.
\State \highlight{$\tmin(v) \leftarrow t$.}
\EndIf
\EndFor
\EndFor
\EndFor
\State \highlight{$\mathcal{U}_j(v) = z^{\tmin(v) a}$ for every $v \in S$.}
\State \highlight{\Return $\Hat{d}(v), \mathcal{U}_j(v)$ for every $v \in S$.}
\end{algorithmic}
\end{algorithm}

\medskip

\noindent We extend Lemma~\ref{lem:nbr_bis} and obtain the following corollary:
\begin{corollary}\label{cor:nbr_unf}
If $v \in S^{ta}$ for some $t \in [T], a \in [\lambda]$, then, for every neighbor $w \in \nbr{v} \cap V \setminus S^{ta}$, we have:
\[ \Pr[w = z^{ta}] = \frac{1}{|\nbr{S^{ta}} \cap V \setminus {S}^{ta}_j|}\]
\end{corollary}
\begin{proof}
From Lemma~\ref{lem:nbr_bis}, we know that any vertex $w \in \nbr{S^{ta}} \cap V \setminus S^{ta}$ will satisfy :
\[ \Pr[w = z^{ta}] = \frac{1}{|\nbr{S^{ta}} \cap V \setminus {S}^{ta}_j|}.\]

As $\nbr{v} \cap V \setminus S^{ta} \subseteq \nbr{S^{ta}} \cap V \setminus S^{ta}$, we have the corollary.
\end{proof}

\begin{lemma}
\label{lem:approxDeg_sampling}
Suppose $S \subseteq V$. Then, Algorithm~\ref{alg:count_sampling} uses $O(\epsilon^{-4}\log^5 n\log (\epsilon^{-1} \log n))$ BIS queries and with probability $1-O(\eps/\log n)$, returns degree estimates $\Hat{d}(v)$ for every vertex $v \in S$ satisfying:
$$ d(v)(1-\epsilon) \leq \Hat{d}(v) \leq d(v) + \frac{\epsilon^4 }{\log^2 n} \cdot d(S). $$
\end{lemma}
\begin{proof}
For every random partition, we use Lemma~\ref{lem:nbr_bis} to return a uniform neighbor. This step uses $O(\log^2 n\log(1/\delta))$ BIS queries and succeeds with probability at least $1-\delta$, where $\delta = O(\epsilon /T\lambda \log n)$, $T = O(\log n)$, and $\lambda = O(\epsilon^{-4} \log^2 n)$. The total number of random partitions considered is $O(T \cdot \lambda) = O(\epsilon^{-4} \log^3 n)$. Following the proof of Lemma~\ref{lem:approxDeg}, we get the lemma.
\end{proof}

\subsection{Identifying a uniform edge in the graph}\label{subsec:sampler}

In this section, we give an algorithm that returns an edge sample from a distribution that is close to the uniform distribution. Our algorithm extends Algorithm~\estimator and is based on the following idea. Suppose we know the degrees of all the vertices, denoted by $d(v) \ \forall v \in V$, in the graph. In order to sample a uniform edge, we can sample a vertex $v$ with probability $d(v)/\sum_{w \in V} d(w)$ and return a uniform neighbor among the neighbors of $v$. We can observe that the probability that an edge $e = (v, u)$ is sampled is $d(v)/\sum_{w \in V} d(w) \cdot 1/d(v) + d(u)/\sum_{w \in V} d(w) \cdot 1/d(u) = 1/m$. 

\begin{algorithm}[!ht]
\caption{\sampling: Non-adaptive algorithm for sampling uniform edges}
\label{alg:estimator_sample}
\begin{algorithmic}[1]
\Statex \textbf{Input:} $V$ set of $n$ vertices and $\epsilon > 0$ error parameter.
\Statex \textbf{Output:} Edge of the graph sampled from a (near)-uniform distribution.
\State Scale $\epsilon \leftarrow \frac{\epsilon}{600\log_{1/\epsilon} \log n}$ and initialize $\gamma \leftarrow {1}/{(1-\epsilon)}$ and $B \leftarrow {2}/{\epsilon}$. 
\State Let $s$ be an integer selected uniformly at random from the interval $[0, B)$.
\State Let $\boundary{j} \leftarrow -s + j \cdot B$ for every integer $j$ in the interval $\left[0, \frac{1}{B} \cdot \log_{\gamma} n\right].$
\State Initialize $S_0 \leftarrow V$ and construct $S_1$ by sampling vertices in $S_0$ with probability $1/\gamma^{\boundary{1}}$.
\State Construct $S_1 \supseteq S_2 ... \supseteq S_L$ for $L = \frac{1}{B} \cdot \log_\gamma n$ where each $S_j$ is obtained by sampling vertices in $S_{j-1} \ \forall j\ge 2$, independently with probability $1/\gamma^B$. 
\For{$j = 0,1,\ldots L$}
\State Run \uneighbor (Algorithm~\ref{alg:count_sampling}) on $S_j$, to obtain the degree estimates $\hat d_j(v)$ satisfying $(1-\eps) d(v) \le \hat d_j(v) \le d(v) + \frac{c_1\epsilon^3 \cdot m}{\log n \cdot \gamma^{\boundary{j}}}$ for all $v \in S_j$.
\State \highlight{Let $\mathcal{U}(v)$ denote the neighbor returned by Algorithm~\ref{alg:count_sampling} for vertex $v \in S_j$.}
\EndFor
\State Let $\bar m_0$ be the $O(\log n)$-approximate estimate from the Algorithm \textsc{CoarseEstimator} in Beame et al.~\cite{beame2020edge} on a random partition of $V$.
\State Set $\bar m_0 \leftarrow \max\{2, 16 \log n \cdot \bar m_0\}$, so that we have $m \le \bar m_0 \le (64\log^2 n) \cdot m.$
\For{$t= 1, 2, \cdots, T = 2\log_{1/\epsilon} \log n$}
\State $\bar m_{t}$ is assigned the output of $\rEstimate$ that takes as input approximate degree values $\hat d_j(v) \ \forall v \in S_j \ \forall j \in [L]$, the previous estimate $\bar m_{t-1}$ and the iteration $t$.
\EndFor
\State \highlight{For every $v$ recovered, let $\hat \ell(v)$ denote the level at which $v$ was recovered by $T$th iteration of \rEstimate. Include all the recovered vertices in $\mathcal{R}.$}
\State \highlight{Let $v_{\textup{sampled}}$ be the vertex drawn from the distribution such that a vertex $v$ in $\RR$ is selected with probability proportional to $\gamma^{\boundary{\hat \ell(v)}} \cdot \hat d_{\hat \ell(v)}(v)$.}
\State \highlight{\Return edge $(v_{\textup{sampled}}, \mathcal{U}(v_{\textup{sampled}}))$.}
\end{algorithmic}
\end{algorithm}

In order to extend the above idea to our setting, there are two challenges. First, we do not know the degrees (or approximate degrees) of all the vertices. This is because the set of recovered vertices, i.e., with $(1\pm \epsilon)$-approximate degree estimates known is a subset of the sampled vertices at each level of sampling. Secondly, each vertex is recovered at a different level and is therefore sampled with different probabilities. In order to return a uniform edge based on the previously discussed idea, we must return a single vertex among the set of recovered vertices with probability proportional to its degree.

We address these two challenges by, amongst the recovered vertices, returning vertex $v$ with probability proportional to $\gamma^{\hat \ell(v)} \cdot \hat d(v)$ where $\hat \ell(v)$ is the level at which it is recovered, and $\hat d_{\ell(v)}(v)$ is the degree estimate at the level of recovery. From Section~\ref{subsec:estimator}, we know that our estimator $\sum_{v \textup{ is recovered}} \hat X(v) = \sum_{v \textup{ is recovered}}\gamma^{\hat \ell(v)} \cdot \hat d(v)$ is concentrated around $2m$ (See Lemma~\ref{lem:approx_counting}) and we will be able to return a \emph{near}-uniform sample.

\medskip

\noindent \textbf{Overview of Algorithm~\sampling}. Our algorithm is an extension of Algorithm~\estimator (Algorithm~\ref{alg:estimator}) and the differences are highlighted in \highlight{blue}. During the process of constructing the edge estimator, by repeated refinements, let $\hat \ell(v)$ denote the level at which a vertex $v$ is recovered at the last, i.e., $T=O(\log_{1/\epsilon} \log n)^{th}$ \textit{refinement} iteration, and the corresponding degree estimate $\hat d_{\hat \ell(v)}(v)$. Let $\RR$ denote the set of all recovered vertices, i.e., $\hat X(v) \neq 0$, in the last \textit{refinement} iteration. Then, we draw a vertex $v$ from the distribution such that it is selected with probability proportional to $\gamma^{\boundary{\hat \ell(v)}} \cdot \hat d_{\hat \ell(v)}(v)$. From our earlier discussion, this approach will result in a \emph{near}-uniform sample.

\subsubsection{Proof of Theorem~\ref{thm:introSampling}}

\begin{claim}\label{cl:degreelevels_low_sampling}
With probability $1-2\epsilon$, for all levels $j \in \{1, 2, \cdots, L\}$, we have:
\[ d(S_j) \le \frac{m \cdot L}{\epsilon \gamma^{\boundary{j}}}.\]
\end{claim}
\begin{proof}

As every vertex is included in $S_j$ with probability $1/\gamma^{\boundary{j}}$, we get: 
\[\E[d(S_j)] = \frac{\sum_{v \in V} d(v)}{\gamma^{\boundary{j}}} = \frac{2m}{\gamma^{\boundary{j}}}\]
Therefore, by Markov's Inequality, $\Pr[d(S_j) \ge {m \cdot L}/{\epsilon \gamma^{\boundary{j}}} ] \le 2\epsilon/L$. Taking a union bound over all the levels, with probability at least $1-2\epsilon$, $$d(S_j) \le {m \cdot L}/\epsilon \gamma^{\boundary{j}} \le {m \cdot \log n}/\epsilon \gamma^{\boundary{j}} \textup{for every level }j \in [L].$$
\end{proof}

By setting $\lambda = O(\epsilon^{-4}\log^2 n)$, a multiplicative factor of $1/\epsilon$ more than that in section~\ref{sec:counting}, we ensure that the exact guarantees hold with probability $1-\epsilon$.  Combining Claim~\ref{cl:degreelevels_low_sampling} and Lemma~\ref{lem:approxDeg_sampling}, for sufficiently large $n$, we have:
\begin{corollary}\label{cor:degree_estimates_sampling}
The degree estimates returned by Algorithm~\ref{alg:count_sampling} for each sampling level $j \in [L]$, satisfy the following with probability at least $1-\epsilon$:
\[ (1-\epsilon) d(v) \le \hat d_j(v) \le d(v) + \frac{c_1 \epsilon^3 m}{\gamma^{\boundary{j}} \log n}  \ \quad \forall v \in S_j. \]
\end{corollary}

For the remaining portion of this section, we will condition on the event that Corollary~\ref{cor:degree_estimates_sampling} is satisfied. In the proof of the main Theorem~\ref{thm:sampling}, we account for the failure probability of this event.

In the next lemma, we show that if a vertex is recovered at level $j$ (recall the threshold value for recovery from Section~\ref{subsec:estimator}), the neighbor returned by Algorithm~\ref{alg:count_sampling}, given by $\mathcal U_j(v)$ is equal to any neighbor of $v$ with probability $1/\hat d(v)$. From Section~\ref{subsec:estimator}, we know that approximate degree of $v$ obtained from the set $S$, denoted by $\hat d_j(v)$, when $\ell(v) = j$ is a $(1\pm \epsilon)$-approximation of $d(v)$, therefore, we have returned a neighbor of $v$ with probability $(1\pm \epsilon)/d(v)$.

\begin{lemma}\label{lem:unf_nbr_sampling}
Suppose $v \in S_{j}$ satisfies the following: $\hat d(v) \ge \frac{m}{\gamma^{\boundary{j}}} \cdot \frac{c_2\epsilon^2}{\log n}$. Then, for every $w \in \nbr{v}$,  we have with probability $1-\epsilon$:
\[  (1-\epsilon) \frac{1}{\hat d(v)} \le \Pr[\mathcal{U}_{j}(v) = w] \le (1+\epsilon) \frac{1}{\hat d(v)}.\]
\end{lemma}
\begin{proof}
From Lemma~\ref{lem:nbr}, we know that:
\[ (1-\epsilon)|\nbr{S^{\tmin(v) a}} \cap V \setminus {S}^{\tmin(v)a}| \le 
    \nbrest{}(S^{\tmin(v)a}) \le (1+\epsilon)|\nbr{S^{\tmin(v)a}} \cap V \setminus {S}^{\tmin(v)a}|.\]
    
Moreover, our degree estimates in Algorithm~\ref{alg:count_sampling} are obtained by $\hat d(v) = \nbrest{}(S^{\tmin(v)a})$. 

Consider a vertex $w \in \nbr{v}$. If $w \in \nbr{v} \setminus S^{\tmin(v) a}$, we have:
\begin{align*}
    \Pr \left[\mathcal{U}_j(v) = w \right] &= \frac{1}{|\nbr{S^{\tmin(v) a}} \cap V \setminus S^{\tmin(v)a}|} \\
    (1-\epsilon) \cdot \frac{1}{\hat d(v)} \le \Pr \left[\mathcal{U}_j(v) = w \right] &\le (1+\epsilon) \cdot \frac{1}{\hat d(v)}
\end{align*}

If $w \in \nbr{v} \cap S^{\tmin(v) a}$, then, it is never returned. In iteration $\tmin(v) \in [T]$, $w \in S$ is assigned to one of the $\lambda$ random partitions, we observe that such an event happens with probability:
\[ \Pr[w \in S^{\tmin(v)a} \cap \nbr{v}] = \frac{1}{\lambda} = \frac{ \epsilon^4}{c\log^2 n} \le \epsilon, \mbox{ for some constant $c > 1$.}\]

It is possible that $\mathcal{U}_j(v) \not \in \nbr{v} \cap (V\setminus S^{\tmin(v)a})$, which is a failure event for us, as no neighbor of $v$ will be returned. We will argue that probability for such an event occurring is small.

From the analysis in Lemma~\ref{lem:approxDeg} and Corollary~\ref{cor:degree_estimates_sampling}, with probability at least $1-1/n^3$ and for an appropriate choice of $c_2$, we have:
\begin{align*}
    d(S^{\tmin(v)a} \setminus \{v\}) &\le \frac{c_1 m  \cdot \epsilon^3}{\gamma^{\boundary{j}} \log n} \le \epsilon \hat d(v) \quad (\mbox{as we are given}\ \hat d(v) \ge \frac{m}{\gamma^{\boundary{j}}} \cdot \frac{c_2\epsilon^2}{\log n}).
\end{align*}

\begin{align*}
    \Pr[\mathcal{U}_j(v) \not \in \nbr{v} \cap (V\setminus S^{\tmin(v)a})] &\le \frac{|\nbr{S^{\tmin(v)a} \setminus \{v\}} \cap V \setminus {S}^{\tmin(v)a}|}{|\nbr{S^{\tmin(v)a}} \cap V \setminus {S}^{\tmin(v)a}|}\\
    &\le  \frac{d(S^{\tmin(v)a} \setminus \{v\})}{|\nbr{S^{\tmin(v)a}} \cap V \setminus {S}^{\tmin(v)a}|}\\
    &\le (1+\epsilon)\cdot \frac{d(S^{\tmin(v)a} \setminus \{v\})}{\hat d(v)} \le \epsilon (1+\epsilon) = 2\epsilon.
\end{align*}

 By union bound, failure probability is at most $1/n^3+\epsilon \le 2\epsilon$. Scaling $\epsilon$ appropriately, gives us the lemma.
 \end{proof}

In the next lemma, we show that if a vertex is not in $\Vboundary$, then, it is recovered with the required probability of $\hat d(v)/2m$ (upto $1\pm \epsilon$ factor). Otherwise, we argue that the probability of returning it is not too large. 
\begin{lemma}\label{lem:vertex_sampled}
For any vertex $v$, with probability at least $1-\epsilon$, we have:
\begin{enumerate}
    \item $ (1-\epsilon) \frac{\hat d_{\hat \ell(v)}(v)}{2m} \le \Pr[v_{\textup{sampled}} = v] \le (1+\epsilon) \frac{\hat d_{\hat \ell(v)}(v)}{2m} \textup{ if } v \not \in \Vboundary.$
    \item $\Pr[v_{\textup{sampled}} = v]\le \frac{15(1+\epsilon)\hat d_{\hat \ell(v)}(v)}{2 m} \textup{ if } v \in \Vboundary.$
\end{enumerate}
\end{lemma}
\begin{proof}
We condition on the event that the $(1-\epsilon) 2m \le \sum_{v \in \RR} \hat X(v) \le (1+\epsilon) 2m $, which happens with probability at least $1-\epsilon$ (follows from Corollary~\ref{cor:degree_estimates_sampling}, Lemma~\ref{lem:approx_counting}, and Theorem~\ref{thm:edge_estimate}).

From Claim~\ref{cl:recovery}, we know that a vertex not lying at the boundary will be recovered at level $\ell(v)$, i.e., $\hat \ell(v) = \ell(v)$. Consider a vertex $v \in V \setminus \Vboundary$, we have:
\[ \Pr[v \in \RR] = \Pr[v \in S_{\hat \ell(v)}] = \frac{1}{\gamma^{\boundary{\hat \ell(v)}}} \]
 
From construction, for any vertex $v \in \RR$, we have that 
\[ \Pr[v_{\textup{sampled}} = v \mid v \in \RR] = \frac{\gamma^{\boundary{\hat \ell(v)}} \hat d_{\hat \ell(v)}(v)}{\sum_{w \in \mathcal R} \gamma^{\boundary{\hat \ell(v)}} \hat d_{\hat \ell(w)}(w)} \le (1+2\epsilon)\frac{\gamma^{\boundary{\hat \ell(v)}} \hat d_{\hat \ell(v)}(v)}{2 m}.\]

Similarly, we get:
\[ \Pr[v_{\textup{sampled}} = v \mid v \in \RR] \ge (1-\epsilon)\frac{\gamma^{\boundary{\hat \ell(v)}} \hat d_{\hat \ell(v)}(v)}{2 m}.\]

Combining both the above statements and scaling $\epsilon = \epsilon/2$, we get:
\[ (1-\epsilon) \frac{\hat d_{\hat \ell(v)}(v)}{2 m} \le \Pr[v_{\textup{sampled}} = v] \le (1+\epsilon) \frac{\hat d_{\hat \ell(v)}(v)}{2 m}.\]

Now, consider a vertex $v \in \Vboundary$. From Claim~\ref{cl:bad_boundary_vertices}, if included in we know that $\hat \ell(v) \in \{\ell(v)-1, \ell(v), \ell(v)+1\}$, provided it is included in their corresponding set $S_{\hat \ell(v)}$.
\begin{align*}
    \Pr[v \in \RR] \le \sum_{\hat \ell(v) \in \{\ell(v)-1, \ell(v), \ell(v)+1\}}\Pr[v \in S_{\hat \ell(v)}] &\le \frac{3}{\gamma^{\boundary{\ell(v)-1}}}\\
    \Pr[v_{\textup{sampled}} = v] = \Pr[v_{\textup{sampled}} = v \mid v \in \RR] \cdot \Pr[v \in \RR] &\le \frac{3(1+2\epsilon)\gamma^{\boundary{\hat \ell(v)}}}{\gamma^{\boundary{\ell(v)-1}}} \cdot \frac{\hat d_{\hat \ell(v)}(v)}{2 m}\\ 
    &\le \frac{15(1+2\epsilon)\hat d_{\hat \ell(v)}(v)}{2 m} 
\end{align*}

Hence, the lemma.
\end{proof}

\begin{theorem}\label{thm:sampling}
Given a graph $G$ with $n$ nodes, $m$ edges, and edge set $E$, there is an algorithm that makes $O(\epsilon^{-4}\log^6 n\log(\epsilon^{-1}\log n) + \epsilon^{-6} \log^5 n \log^{6} (\log n) \log(\epsilon^{-1} \log n) )$ non-adaptive BIS queries which, with probability at least $1-\epsilon$, outputs an edge from a probability distribution $P$ satisfying   $(1-\epsilon)/{m} \le P(e) \le (1+\epsilon)/{m}$ for every $e \in E$.
\end{theorem}
\begin{proof}
Consider an edge $e = (v, u)$. Then, the edge $e = (v, u)$ can be returned by Algorithm~\ref{alg:estimator_sample} if either $v$ or $u$ is the vertex sampled $v_{\textup{sampled}}$ and the other vertex is the neighbor returned by Algorithm~\ref{alg:count_sampling}. From Lemmas~\ref{lem:unf_nbr_sampling} and~\ref{lem:vertex_sampled}, we have:
\begin{align*}
    \Pr[e \textup{ is returned by Algorithm~\ref{alg:estimator_sample}}] &= \Pr[v_{\textup{sampled}} = v] \cdot \Pr[\mathcal{U}(v) = u] + \Pr[v_{\textup{sampled}} = u] \cdot \Pr[\mathcal{U}(u) = v] \\
    \Rightarrow \Pr[v_{\textup{sampled}} = v] &= \Pr[v_{\textup{sampled}} = v \mid v \in V \setminus \Vboundary] \Pr[v \in V \setminus \Vboundary] \\
    &+ \Pr[v_{\textup{sampled}} = v \mid v \in \Vboundary] \Pr[v \in \Vboundary]\\
    &\le (1+\epsilon)^2 \frac{\hat d_{\hat \ell(v)}(v)}{2 m} + \epsilon \cdot \frac{15(1+2\epsilon)\hat d_{\hat \ell(v)}(v)}{2 m}\\
    &\le (1+O(\epsilon)) \cdot \frac{\hat d_{\hat \ell(v)}(v)}{2 m}\\
\end{align*}
Therefore, we have:
\begin{align*}
     \Pr[e \textup{ is returned by Algorithm~\ref{alg:estimator_sample}}] &\le (1+O(\epsilon)) \cdot \frac{\hat d_{\hat \ell(v)}(v)}{2 m} \cdot \frac{1}{\hat d_{\hat \ell(v)}(v)} + (1+O(\epsilon)) \cdot \frac{\hat d_{\hat \ell(u)}(u)}{2 m} \cdot \frac{1}{\hat d_{\hat \ell(u)}(u)} \\
    &\le (1+O(\epsilon)) \cdot \frac{1}{m}.
\end{align*}
The total number of additional (other than those used for edge estimation) BIS queries used is $O(\log n \cdot Q)$ where $Q$ is the queries used by Algorithm~\ref{alg:count_sampling} to return a (near) uniform sample. From Claim~\ref{lem:approxDeg_sampling}, we have that $Q = O(\epsilon^{-4}\log^5 n\log(\epsilon^{-1}\log n))$. In Algorithm~\ref{alg:count_sampling}, as we partition each sampled subset $S_j$, for every $j \in L$, into an additional $1/\epsilon$ factor many partitions as compared to Algorithm~\ref{alg:count}. Therefore, we use a total of $O(\epsilon^{-4}\log^6 n\log(\epsilon^{-1}\log n) + \epsilon^{-6} \log^5 n \log^{6} (\log n) \log(\epsilon^{-1} \log n) )$ BIS queries for edge estimation. 

Using union bound, we have that the failure probability in Lemma~\ref{lem:vertex_sampled} and Corollary~\ref{cor:degree_estimates_sampling} is at most $O(\epsilon)$. Scaling the $\epsilon$ appropriately, gives us a failure probability of $\epsilon$. Hence, the theorem.
\end{proof}

\section{Graph Connectivity}\label{sec:graph_connectivity}
In this section, we present Algorithm~\textsc{Connectivity-BIS} that uses $2$-rounds of adaptivity to determine the connectivity of an input graph $G$. This improves upon on a prior three-round algorithm of~\cite{assadi2021}. In particular, the algorithm of~\cite{assadi2021} selects $O(\log^2 n)$ random neighbors per vertex, and contracts the connected components of this random graph into \emph{supernodes}. This random sampling step can be performed using one round of $\tilde O(n)$ BIS queries. They prove that in the contracted graph on the supernodes, there are at most $O(n\log n)$ edges. Using this fact, they then show how to identify whether all the supernodes are connected using $\tilde O(n)$ BIS queries and two additional around of adaptivity.

We follow the same basic approach: using a first round of $\tilde O(n)$  queries to randomly sample $O(\log^2 n)$ neighbors per vertex and contract the graph into supernodes. Once this is done, we observe that we have BIS query access to the contracted graph simply by always grouping together the set of nodes in each supernode. So, we can directly apply the non-adaptive sampling algorithm of Theorem~\ref{thm:introSampling} to sample edges from the contracted graph. By a coupon collecting argument, drawing $O(n\log^2 n)$ near-uniform edge samples (with replacement) from the contracted graph suffices to recover all $O(n \log n)$ edges in the graph, and thus determine connectivity of the contracted graph, and, in turn, the original graph. 

\medskip
\noindent \textbf{Algorithm}~\textsc{Connectivity-BIS}:
\begin{enumerate}
    \item For every node $v \in V$, sample $O(\log^2 n)$ edges uniformly with replacement, from the neighborhood $\nbr{v}$, using Lemma~\ref{lem:nbr_bis}. Let the resulting set of edges sampled be denoted by $E' \subseteq E$ and the connected components in the subgraph $G(V, E')$ be $S_1, S_2, \cdots S_p$.
    \item Let $\Gsup(\Vsup, \Esup)$ where $\Esup \subseteq \Vsup \times \Vsup$ denotes the supergraph obtained from $G(V, E')$ by collapsing the connected components $S_1, S_2, \cdots S_p$ into single supernodes $s_1, s_2, \cdots s_p$ respectively, given by: $$\Vsup = \{ s_i \mid S_i \text{ where }i\in [p] \text{ is a connected component in }G(V, E') \}$$ $$\Esup = \{ (s_i, s_j) \mid \exists x \in S_i, y \in S_j \text{ where } i \neq j \text{ such that } (x, y) \in E \}.$$  
    \item Run Algorithm~\ref{alg:estimator_sample} (with any constant value for $\epsilon$) on $\Gsup$ to draw $T = O(n\log^2 n)$ uniform superedge samples with replacement, from $\Esup$ . If the resulting graph $\Gsup$ is connected, output `Yes'. Otherwise, output `No'.
\end{enumerate}

\begin{theorem}\label{thm:graph_connectivity}
Given a graph $G$ with $n$ nodes, there is a $2$-round adaptive algorithm that determines if $G$ is connected with probability at least $1-1/n$ using $\Tilde{O}(n\log^{8} n)$ BIS queries,  where $\Tilde{O}(\cdot)$ ignores the $\log^{O(1)} \log n$ dependencies.
\end{theorem}
\begin{proof}
We have: $|\Esup| = O(n\log n)$ (see Lemma 6.5 in \cite{assadi2021}). From Theorem~\ref{thm:sampling}, we have:
\begin{align*}
    \Pr[(s_i, s_j) \in \Esup \mbox{ is returned}] &\ge (1-\epsilon) \cdot \frac{(1 - \epsilon)}{|\Esup|} \ge (1-2\epsilon) \cdot c \cdot \frac{1}{n\log n} \mbox{ for constant $\epsilon < 0.5$}\\
    \Rightarrow \Pr[(s_i, s_j) \mbox{ is not returned}] &= \left(1-\Pr[(s_i, s_j) \in \Esup \mbox{ is returned}]\right)^T\\
    &\le e^{-\frac{2 c}{3} \cdot \frac{1}{n\log n} \cdot T}
    \le \frac{1}{n^4}.
\end{align*}
By union bounding over at most $O(n\log n)$ many superedges, the total failure probability is at most $1/n^2$. Similarly, union bounding over the failure probability of recovering $O(n \log^2 n)$ edges in the first step, we have that the failure probability is at most $1/n^2$. Therefore, Algorithm~\textsc{Connectivity-BIS} recovers all the edges in $\Esup$ with probability at least $1-1/n$.

From Lemma~\ref{lem:nbr_bis}, the total number of BIS queries required in the first round of our algorithm is $O(n \cdot \log^2 n \cdot \log^3 n) = O(n \log^5 n)$. Setting $\epsilon$ to be any constant value, from Theorem~\ref{thm:sampling}, the total number of BIS queries required is $O(n \log^5 n) + \Tilde{O}(n\log^2 n \cdot \log^{6} n) = \Tilde{O}(n \log^{8} n)$, where $\Tilde{O}(\cdot)$ ignores the $\log \log n$ dependencies. Hence, the theorem.
\end{proof}

\section{Conclusion and Open Questions}\label{sec:conclusion}
In this paper, we presented non-adaptive algorithms for edge estimation and sampling using BIS queries. It would be interesting to see if there is a better dependence on $\epsilon$ than that obtained by our algorithms, when we consider non-adaptive algorithms for edge estimation. Using Independent Set (IS) queries, \textit{adaptive} algorithms for edge estimation with optimal query complexity $O(\min\{ \sqrt{m}, n/\sqrt{m} \} \cdot \textup{poly}(\log n, 1/\epsilon))$ were obtained only recently~\cite{chen2020nearly, beame2020edge}. It would be interesting to see if we can extend our techniques to study \emph{non-adaptive} algorithms for edge estimation using IS queries or in the standard adjacency list model. We believe that adaptivity plays a role similar to that of number of passes for streaming algorithms, and optimizing for the rounds of adaptivity could be a good future direction, even for problems that might already have query optimal adaptive sub-linear time algorithms.

\section*{Acknowledgements}
Part of this work was done while R. Addanki was a visiting student at the Simons Institute for the Theory of Computing. This work was supported by a Dissertation Writing Fellowship awarded by the Manning College of Information and Computer Sciences, University of Massachusetts Amherst to R. Addanki. In addition, this work was supported by NSF grants CCF-1934846, CCF-1908849, and CCF-1637536, awarded to A. McGregor; and NSF grants CCF-2046235, IIS-1763618, as well as Adobe and Google Research Grants, awarded to C. Musco. We thank the anonymous reviewers of the European Symposium on Algorithms (ESA)  2022, for their helpful suggestions.

\bibliographystyle{alpha}
\bibliography{references}

\end{document}